\def\llncs{0}
\def\fullpage{1}
\def\anonymous{0}
\def\authnote{1}
\def\notxfont{0}
\def\submission{0}

\ifnum\submission=1
\def\llncs{1}
\fi

\ifnum\llncs=1
\documentclass[envcountsect,a4paper,runningheads,10pt]{llncs}
\else
	\documentclass[letterpaper,hmargin=1.05in,vmargin=1.05in,11pt]{article}
			\ifnum\fullpage=1
		\usepackage{fullpage}
		\fi
\fi

\usepackage[%
  colorlinks=true,
  citecolor=blue,
  pagebackref=true
]{hyperref}

\usepackage{amsmath, amsfonts, amssymb, mathtools,amscd}

\usepackage{amsthm}

\usepackage{lmodern}
\usepackage[T1]{fontenc}
\usepackage[utf8]{inputenc}

\usepackage{etex}

\usepackage{arydshln} 
\usepackage{url}
\usepackage{ifthen}
\usepackage{bm}
\usepackage{multirow}
\usepackage[dvips]{graphicx}
\usepackage[usenames]{color}
\usepackage{xcolor,colortbl} 
\usepackage{threeparttable}
\usepackage{comment}
\usepackage{paralist,verbatim}
\usepackage{cases}
\usepackage{booktabs}
\usepackage{braket}
\usepackage{cancel} 
\usepackage{ascmac} 
\usepackage{framed}
\usepackage{authblk}
\usepackage{pifont}
\usepackage{qcircuit}
\usepackage{tikz}
\usetikzlibrary{cd}
\definecolor{darkblue}{rgb}{0,0,0.6}
\definecolor{darkgreen}{rgb}{0,0.5,0}
\definecolor{maroon}{rgb}{0.5,0.1,0.1}
\definecolor{dpurple}{rgb}{0.2,0,0.65}

\usepackage[capitalise,noabbrev]{cleveref}
\usepackage[absolute]{textpos}
\usepackage[final]{microtype}
\usepackage[absolute]{textpos}
\usepackage{everypage}
\DeclareMathAlphabet{\mathpzc}{OT1}{pzc}{m}{it}

\usepackage{algorithmic}
\usepackage{algorithm}
\usepackage{here}

\usepackage[normalem]{ulem}

\newtheoremstyle{thicktheorem}%
{\topsep}
{\topsep}
{\itshape}{}%
{\bfseries}%
{.}
{ }%
{\thmname{#1}\thmnumber{ #2}%
		\thmnote{ (#3)}%
}

\newtheoremstyle{remark}
{\topsep}
{\topsep}
	{}
	{}
	{}
	{.}
	{ }
	{\textit{\thmname{#1}}\thmnumber{ #2}
			\thmnote{ (#3)}%
	}

\ifnum\llncs=0
	\theoremstyle{thicktheorem}
	\newtheorem{theorem}{Theorem}[section]
	\newtheorem{lemma}[theorem]{Lemma}
	\newtheorem{corollary}[theorem]{Corollary}
	
	\newtheorem{definition}[theorem]{Definition}

	\theoremstyle{remark}
	
	\newtheorem{remark}[theorem]{Remark}

\else
\fi

\Crefname{MyClaim}{Claim}{Claims}

	\crefname{theorem}{Theorem}{Theorems}
	\crefname{assumption}{Assumption}{Assumptions}
	\crefname{construction}{Construction}{Constructions}
	\crefname{corollary}{Corollary}{Corollaries}
	\crefname{conjecture}{Conjecture}{Conjectures}
	\crefname{definition}{Definition}{Definitions}
	\crefname{exmaple}{Example}{Examples}
	\crefname{experiment}{Experiment}{Experiments}
	\crefname{counterexample}{Counterexample}{Counterexamples}
	\crefname{lemma}{Lemma}{Lemmata}
	\crefname{observation}{Observation}{Observations}
	\crefname{proposition}{Proposition}{Propositions}
	\crefname{remark}{Remark}{Remarks}
	\crefname{claim}{Claim}{Claims}
	\crefname{fact}{Fact}{Facts}
	\crefname{note}{Note}{Notes}

\ifnum\llncs=1
 \crefname{appendix}{App.}{Appendices}
 \crefname{section}{Sec.}{Sections}
\else
\fi

\ifnum\llncs=1
\renewcommand*{\backref}[1]{}
\else
	\renewcommand*{\backref}[1]{(Cited on page~#1.)}
	\ifnum\notxfont=1
	\else
		\usepackage{newtxtext}
	\fi
\fi
\ifnum\authnote=0  
\newcommand{\mor}[1]{}
\newcommand{\shogo}[1]{}
\newcommand{\takashi}[1]{}
\newcommand{\fuyuki}[1]{}

\else
\newcommand{\mor}[1]{$\ll$\textsf{\color{red} Tomoyuki: { #1}}$\gg$}
\newcommand{\takashi}[1]{$\ll$\textsf{\color{orange} Takashi: { #1}}$\gg$}
\newcommand{\shogo}[1]{$\ll$\textsf{\color{darkgreen} Shogo: { #1}}$\gg$}
\newcommand{\fuyuki}[1]{$\ll$\textsf{\color{darkblue} Fuyuki: { #1}}$\gg$}

\DeclareRobustCommand{\Erase}{\bgroup\markoverwith{\textcolor{red}{\rule[.5ex]{2pt}{0.4pt}}}\ULon}

\fi

\newcommand{\Tr}{\mathrm{Tr}}

\newcommand{\StateGen}{\mathsf{StateGen}}
\newcommand{\Tag}{\mathsf{Tag}}
\newcommand{\Mint}{\mathsf{Mint}}

\newcommand{\cA}{\mathcal{A}}
\newcommand{\cB}{\mathcal{B}}
\newcommand{\cC}{\mathcal{C}}

\newcommand{\cE}{\mathcal{E}}
\newcommand{\cF}{\mathcal{F}}

\newcommand{\cX}{\mathcal{X}}
\newcommand{\cY}{\mathcal{Y}}

\def\makeuppercase#1{
\expandafter\newcommand\csname tl#1\endcsname{\widetilde{#1}}
}

\def\makelowercase#1{
\expandafter\newcommand\csname tl#1\endcsname{\widetilde{#1}}
}

\newcommand{\regC}{\mathbf{C}}

\newcommand{\regR}{\mathbf{R}}
\newcommand{\regZ}{\mathbf{Z}}

\newcommand{\regM}{\mathbf{M}}

\newcommand{\regB}{\mathbf{B}}
\newcommand{\regA}{\mathbf{A}}
\newcommand{\regX}{\mathbf{X}}
\newcommand{\regY}{\mathbf{Y}}

\newcommand{\secp}{\lambda}

\newcommand{\A}{\entity{A}}

\newcommand*{\sk}{\keys{sk}}
\newcommand*{\pk}{\keys{pk}}

\newcommand*{\sigk}{\keys{sigk}}

\newcommand{\ct}{\keys{ct}}

\newcommand*{\vk}{\keys{vk}}

\newcommand*{\msg}{\keys{msg}}

\newcommand*{\keys}[1]{\mathsf{#1}}

\newcommand*{\algo}[1]{\ensuremath{\mathsf{#1}}}

\newcommand*{\entity}[1]{\mathcal{#1}}

\newenvironment{boxfig}[2]{\begin{figure}[#1]\fbox{\begin{minipage}{0.97\linewidth}
                        \vspace{0.2em}
                        \makebox[0.025\linewidth]{}
                        \begin{minipage}{0.95\linewidth}
            {{
                        #2 }}
                        \end{minipage}
                        \vspace{0.2em}
                        \end{minipage}}}{\end{figure}}

\newcommand{\bit}{\{0,1\}}

\newcommand{\Gen}{\algo{Gen}}
\newcommand{\KeyGen}{\algo{KeyGen}}
\newcommand{\SKGen}{\algo{SKGen}}
\newcommand{\PKGen}{\algo{PKGen}}

\newcommand{\Enc}{\algo{Enc}}

\newcommand{\Dec}{\algo{Dec}}

\newcommand{\Sign}{\algo{Sign}}

\newcommand{\Ver}{\algo{Ver}}

\newcommand{\st}{\algo{st}}

\newcommand{\PRF}{\algo{PRF}}

\newcommand{\Eval}{\algo{Eval}}

\newcommand{\negl}{{\mathsf{negl}}}

\newcommand{\poly}{{\mathrm{poly}}}

\DeclareMathOperator*{\Exp}{\mathbb{E}}

\usetikzlibrary{decorations.markings}
\tikzset{
  cross/.style={
    postaction={decorate,decoration={markings,
    mark=at position 0.45 with {\draw[-,line width=1pt] (-10pt,-10pt) -- (10pt,10pt);\draw[-,line width=1pt] (-10pt,10pt) -- (10pt,-10pt);}}}
  }
}

\makeatletter
\DeclareRobustCommand
  \myvdots{\vbox{\baselineskip4\p@ \lineskiplimit\z@
    \hbox{.}\hbox{.}\hbox{.}}}
\makeatother

\title{
Quantum Unpredictability
}

\ifnum\anonymous=1
\ifnum\llncs=1
\author{\empty}\institute{\empty}
\else
\author{}
\fi
\else
\ifnum\llncs=1
\author{
Tomoyuki Morimae\inst{1} \and Takashi Yamakawa\inst{1,2} 
}
\institute{
	Yukawa Institute for Theoretical Physics, Kyoto University, Kyoto, Japan \and NTT Social Informatics Laboratories, Tokyo, Japan 
}
\else
\author[1]{Tomoyuki Morimae}
\author[1]{Shogo Yamada}
\author[2,3,1]{Takashi Yamakawa}
\affil[1]{{\small Yukawa Institute for Theoretical Physics, Kyoto University, Kyoto, Japan}\authorcr{\small tomoyuki.morimae@yukawa.kyoto-u.ac.jp}\authorcr{\small shogo.yamada@yukawa.kyoto-u.ac.jp}}
\affil[2]{{\small NTT Social Informatics Laboratories, Tokyo, Japan}\authorcr{\small 
takashi.yamakawa@ntt.com}}
\affil[3]{\small NTT Research Center for Theoretical Quantum Information, Atsugi, Japan}

\fi 
\fi

\date{}

\begin{document}

\maketitle

\begin{abstract}
Unpredictable functions (UPFs) play essential roles in classical cryptography, including message authentication codes (MACs) and digital signatures. 
In this paper, we introduce a quantum analog of
UPFs, which we call unpredictable state generators (UPSGs).
UPSGs are implied by 
pseudorandom function-like states generators (PRFSs), which are a quantum analog of pseudorandom functions (PRFs),
and therefore UPSGs could exist even if one-way functions do not exist, similar to other recently introduced primitives like pseudorandom state generators (PRSGs),
one-way state generators (OWSGs), and EFIs.
In classical cryptography, UPFs are equivalent to PRFs, but 
in the quantum case, the equivalence is not clear, and UPSGs could be weaker than PRFSs.
Despite this, we demonstrate that all known applications of PRFSs are also achievable with UPSGs. 
They include
IND-CPA-secure secret-key encryption and EUF-CMA-secure MACs with unclonable tags. 
Our findings suggest that, for many applications, 
quantum unpredictability,
rather than quantum pseudorandomness,
is sufficient.
\end{abstract}

\ifnum\submission=1
\else
\clearpage
\newpage
\setcounter{tocdepth}{2}
\tableofcontents
\fi
\newpage

\section{Introduction}
\label{sec:introduction}

\subsection{Background}
Pseudorandom functions (PRFs), first formalized
by Goldreich, Goldwasser and Micali in 1984~\cite{JACM:GolGolMic86}, are one of the most fundamental primitives in classical cryptography. 
A PRF is an efficiently-computable keyed function that is computationally indistinguishable from a random function 
for any polynomial-time adversary that can query the function.
PRFs have many important applications in cryptography, and in particular, they are essential building blocks of
EUF-CMA-secure message authentication codes (MACs)
and IND-CPA-secure secret-key encryption (SKE).

Naor and Reingold~\cite{C:NaoRei98} introduced a related primitive so-called unpredictable functions (UPFs).
Like PRFs, a UPF is an efficiently-computable keyed function, but the crucial difference is that the goal of the adversary is not to distinguish it
from the random function, but
to predict the output corresponding to an input that was not queried before.
More precisely, let $f\coloneqq\{f_k\}_k$ be an efficiently-computable keyed function. 
Then $f$ is a UPF if it satisfies the following property, which is called unpredictability:
\begin{align}
\Pr[y=f_k(x):k\gets\bit^\secp,(x,y)\gets\cA^{f_k(\cdot)}]\le\negl(\secp)    
\end{align}
for any polynomial-time adversary $\cA$,
where $x$ was not queried by $\cA$.
It is easy to see that PRFs imply UPFs.
The other direction is not straightforward, but
Naor and Reingold showed that UPFs imply PRFs~\cite{C:NaoRei98},
and therefore PRFs and UPFs are actually equivalent.

What happens if we consider quantum versions of PRFs and UPFs?
Recently, quantum analogs of elementary primitives, including one-way functions (OWFs), pseudorandom generators (PRGs), and 
PRFs, 
have been extensively studied~\cite{C:JiLiuSon18,C:MorYam22,C:AnaQiaYue22,ITCS:BraCanQia23,TCC:AGQY22,AC:Yan22,cryptoeprint:2022/1336,cryptoeprint:2023/543,cryptoeprint:2023/904,nonadaptivePRU}.
For example, pseudorandom states generators (PRSGs) introduced by Ji, Liu, and Song~\cite{C:JiLiuSon18} are a quantum analog of PRGs.
One-way states generators (OWSGs) introduced by Morimae and Yamakawa~\cite{C:MorYam22} are a quantum analog of OWFs.
EFIs introduced by Brakerski, Canetti, and Qian~\cite{ITCS:BraCanQia23} are a quantum analog of EFID~\cite{Gol90}.\footnote{An EFID is a pair of two efficiently samplable classical distributions that are
statistically far but computationally indistinguishable. An EFI is its quantum analog: a pair of two efficiently generatable quantum states that are statistically far but computationally indistinguishable.}
There are mainly two reasons why studying such new quantum elementary primitives
are important.
First, they could be weaker than (quantumly-secure) OWFs~\cite{Kre21,kretschmer2023quantum}, which are the most fundamental assumption in classical cryptography.
More precisely, even if $\mathbf{BQP}=\mathbf{QMA}$ or $\mathbf{P}=\mathbf{NP}$ and therefore OWFs do not exist,
these new primitives could exist (relative to oracles).
Second, despite that, they have many useful applications, such as private-key quantum money, SKE, non-interactive commitments, 
digital signatures, and multiparty computations, etc.
These facts suggest that these primitives will play the role of the most fundamental assumptions in quantum cryptography,
similar to OWFs in classical cryptography.

Quantum versions of PRFs were already studied. There are two quantum analogs of PRFs.
One is pseudorandom unitary operators (PRUs) that were introduced by Ji, Liu, and Song~\cite{C:JiLiuSon18}.\footnote{Weaker variants, so-called pseudorandom states scramblers~\cite{PRSSs} and pseudorandom isometries~\cite{pseudorandom_isometries} were recently introduced. They are shown to be constructed from
OWFs. }
It is a set $\{U_k\}_k$ of efficiently implementable unitary operators that are computationally indistinguishable from Haar random unitary operators.
The other quantum analog of PRFs is
pseudorandom function-like states (generators) (PRFSs) that were introduced by Ananth, Qian and Yuen~\cite{C:AnaQiaYue22}.
A PRFS is a QPT algorithm that, on input a secret key $k$ and a classical bit string $x$, outputs a
quantum state $\phi_k(x)$.
The security roughly means that no QPT adversary can tell whether it is querying to the PRFS oracle or to the oracle that returns Haar random states.\footnote{If the query $x$ was not queried before, the 
oracle samples a new Haar random state $\psi_x$ and outputs it. If the query $x$ was done before, the oracle outputs the same $\psi_x$ that was sampled before.}
EUF-CMA-secure MACs (with quantum tags) and IND-CPA-secure SKE (with quantum ciphertexts) can be constructed from PRFSs~\cite{C:AnaQiaYue22}.

On the other hand, no quantum analog of UPFs was explored before.
Is it equivalent to a quantum analog of PRFs, such as PRUs or PRFSs?
Does it imply 
EUF-CMA-secure MACs and IND-CPA-secure SKE like
PRFSs and PRUs?
Can we gain any meaningful insight for quantum cryptography by studying it?

\subsection{Our Results}
The goal of the present paper is to initiate the study of a quantum version of UPFs which we call unpredictable state generators (UPSGs).
We define UPSGs and construct several cryptographic applications from UPSGs.
UPSGs are implied by PRFSs, and therefore UPSGs could exist even if OWFs do not exist, similar to PRSGs, OWSGs, and EFIs.
As we will explain later, the equivalence between PRFSs and UPSGs are not clear,
and UPSGs could be weaker than PRFSs.
Despite this, we show that all known applications of PRFSs are also achievable with UPSGs.\footnote{Strictly speaking, MACs with unclonable tags 
that are realized with PRFSs satisfy the security against QPT adversaries that query the oracle {\it quantumly}, but those realized
with UPSGs satisfy that only for the classical oracle query.}
This finding provides us with an insightful observation: 
{\it For many applications,
quantum unpredictability, rather than
quantum pseudorandomness, 
is sufficient.}

Relations among our results and known results are summarized in \cref{fig:primitives}.

\paragraph{Defining UPSGs.} 
Our first contribution is to define UPSGs.
A UPSG is a QPT algorithm $\Eval$ that, on input a secret key $k$ and a classical bit string $x$, outputs a quantum state $\phi_k(x)$.
Intuitively, the security (unpredictability) is as follows: no QPT adversary, which can query the oracle $\Eval(k,\cdot)$, can output $(x^*,\rho)$ such that
$x^*$ was not queried and $\rho$ is close to $\phi_k(x^*)$.\footnote{We could consider classical query or
quantum query. In the latter case, it is not clear what we mean by ``not queried''. 
One possible formalization, which we actually adopt, is to define that a bit string $x$
was not queried if the weight of $\ket{x}$ is zero for all quantum queries. For more precise statements, see \cref{sec:definition_UPSG}.}

In the classical case, PRFs and UPFs are equivalent~\cite{C:NaoRei98}.
What happens in the quantum case?
In fact, we can show that PRFSs imply UPSGs. However,
the other direction is not clear.
In the classical case, the construction of PRFs from UPFs is done
by using the Goldreich-Levin \cite{C:NaoRei98,STOC:GolLev89}:
if $f_k(\cdot)$ is a UPF, $g_{k,r}(x)\coloneqq f_k(x)\cdot r$ is a PRF with the key $(k,r)$, where $x\cdot y$ is the inner product between bit strings $x$ and $y$.
However, we cannot directly apply that idea to UPSGs: In particular, what is $\phi_k(x)\cdot r$?

In summary, a quantum analog of UPFs, UPSGs, are implied by PRFSs, 
which especially means that UPSGs could also exist even if OWFs do not exist.
However, the equivalence is not clear, and UPSGs could be weaker than PRFSs.
Then, a natural question is the following: Do UPSGs have useful applications like PRFSs?

\paragraph{IND-CPA-secure SKE.}
Our second contribution is 
to construct IND-CPA-secure SKE (with quantum ciphertexts) from UPSGs.
In the classical case, unpredictability implies pseudorandomness~\cite{C:NaoRei98}, which implies encryption.
However, in the quantum case, as we have explained before, we do not know 
how to convert unpredictability to pseudorandomness,
and therefore it is not self-evident whether SKE can be constructed from UPSGs.
Despite this,
we show that it is actually possible:
\begin{theorem}
\label{thm:SKE}
If UPSGs exist, then IND-CPA-secure SKE exist.    
\end{theorem}

IND-CPA-secure SKE can be constructed from PRFSs~\cite{C:AnaQiaYue22}.
\cref{thm:SKE} shows that such SKE can be constructed from a possibly weaker primitive, UPSGs.

\paragraph{MACs with unclonable tags.}
Our third contribution is
to define and
construct EUF-CMA-secure MACs with unclonable tags from UPSGs.\footnote{We will see that the unclonability of tags automatically implies EUF-CMA security, and therefore we have only to focus on the unclonability of tags.}
The unclonability of tags roughly means that no QPT adversary can, given $t$-copies of a quantum tag, output a large (possibly entangled) quantum state that
contains at least $t+1$ valid tag states.
MACs with unclonable tags are useful in practical applications. For example, consider the
following attack (which is known as the {\it replay attack} in the classical cryptography):
Alice sends the message ``transfer \$100 to Bob'' with a MAC tag to a bank. 
Malicious Bob can steal the pair of the message and the tag, and sends it ten times to the bank so that he can get \$1000. 
In the classical cryptography, the standard EUF-CMA security of MACs cannot avoid such an attack, and
some higher-level treatments are necessary. For example, common techniques are using counters or time-stamps,
but they require the time synchronization among users.

If tags are unclonable, we can avoid such a replay attack.
Actually, it is easy to see that UPSGs imply EUF-CMA-secure MACs with quantum tags. (We have only to take $\phi_k(x)$ as the tag of the message $x$.)
However, the mere fact that tags are quantum does not automatically imply the unclonability of tags.
Moreover, it is not self-evident whether the quantum unpredictability implies unclonability.
(Quantum pseudorandomness implies unclonability~\cite{C:JiLiuSon18}, but it is not clear whether a possibly weaker notion of quantum unpredictability
also implies unclonability.)
Despite that, we show that MACs with unclonable tags can be constructed from UPSGs.
\begin{theorem}
\label{thm:MAC}
If UPSGs exist, then EUF-CMA-secure MACs with unclonable tags exist.    
\end{theorem}
EUF-CMA-secure MACs with unclonable tags can be constructed from PRFSs~\cite{C:AnaQiaYue22}.\footnote{\cite{C:AnaQiaYue22} only showed that PRFSs imply
EUF-CMA-secure MACs with quantum tags, but we can easily show that tags are actually unclonable because their tags are pseudorandom.}
\cref{thm:MAC} shows that EUF-CMA-secure MACs with unclonable tags can be constructed from a possibly weaker primitive, UPSGs.\footnote{Strictly speaking,
there is a difference: MACs with unclonable tags that are realized with PRFSs satisfy the security against QPT adversaries that query the oracle {\it quantumly},
but those realized with UPSGs satisfy only the security against the classical query.}

\paragraph{Private-key quantum money.}
The definition of MACs with unclonable tags straightforwardly implies that of private-key quantum money schemes in \cite{C:JiLiuSon18}.
We therefore have the following as a corollary of \cref{thm:MAC}. ( For the definition of private-key quantum money schemes and a proof of \cref{coro:QM_from_UPSG}, see \cref{sec:QM_from_UPSG}.)
\begin{corollary}\label{coro:QM_from_UPSG}
If UPSGs exist, then private-key quantum money schemes exist.
\end{corollary}

\paragraph{OWSGs and EFIs.}
IND-CPA-secure SKE implies one-time-secure SKE,
and one-time-secure SKE implies OWSGs and EFIs~\cite{cryptoeprint:2022/1336}.
We therefore have the following as a corollary of \cref{thm:SKE}.
\begin{corollary}
If UPSGs exist, then OWSGs and EFIs exist.
\end{corollary}

However, thus obtained OWSGs are mixed OWSGs (i.e., the ones with mixed states outputs), because ciphertexts of the SKE from UPSGs
are mixed states. We can actually directly show that UPSGs imply pure OWSGs:
\begin{theorem}
If UPSGs exist, then pure OWSGs exist.  
\end{theorem}

Because pure OWSGs are broken if ${\bf PP}={\bf BQP}$~\cite{cavalar2023computational},
we also have the following corollary:
\begin{corollary}
If UPSGs exist, then ${\bf PP}\neq{\bf BQP}$.    
\end{corollary}

\subsection{Technical Overview}

\paragraph{IND-CPA-secure SKE from UPSGs.}
Let us first recall a construction of IND-CPA-secure SKE from UPFs in classical cryptography.
In the classical case, we first use the Goldreich-Levin \cite{STOC:GolLev89} to construct PRFs from UPFs:
Let $f_k(\cdot)$ be a UPF. Then $g_{k,r}(x)\coloneqq f_k(x)\cdot r$ is a PRF with the key $(k,r)$~\cite{C:NaoRei98}.
With a PRF $F_k(\cdot)$, an IND-CPA-secure SKE scheme can be constructed as follows:
The secret key is the key of the PRF.
The ciphertext of a message $m$ is $\ct=(r,F_k(r)\oplus m)$ with a random bit string $r$.

However, a similar strategy does not work in the quantum case.
In particular, we do not know how to convert UPSGs to PRFSs: what is $\phi_k(x)\cdot r$!?

Our idea is to use the duality between the swapping and the distinction~\cite{AAS20,EC:HhaMorYam23,kitagawa2023quantum}. The duality
intuitively means that distinguishing two orthogonal states $\ket{\psi}$ and $\ket{\phi}$ is as hard as
swapping $\ket{\psi}+\ket{\phi}$ and $\ket{\psi}-\ket{\phi}$ with each other.
Our ciphertext for a single bit message $b\in\bit$ is, then,
$\ct_b\coloneqq(x,y,\ket{ct_{x,y}^b})$, where
$\ket{ct_{x,y}^b}\coloneqq|0\|x\rangle|\phi_k(0\|x)\rangle+(-1)^b|1\|y\rangle|\phi_k(1\|y)\rangle$, and
$x$ and $y$ are random bit strings. 
Here, $|\phi_k(0\|x)\rangle$
and $|\phi_k(1\|y)\rangle$ are outputs of UPGSs on inputs $0\|x$ and $1\|y$, respectively.
The secret key of our SKE scheme is the key $k$ of the UPSGs.
If a QPT adversary can distinguish $\ct_0$ and $\ct_1$, then due to the duality, we can construct another QPT adversary that can
convert $|\phi_k(0\|x)\rangle$
to $|\phi_k(1\|y)\rangle$. However, it contradicts the unpredictability of the UPSGs.

This argument seems to work. There is, however, one subtle issue here.
The adversary of the IND-CPA security can query the encryption oracle,
but in general we do not know whether the duality works if the distinguisher queries to an oracle,
because the swapping unitary is constructed from the distinguishing unitary and its inverse.

We can solve the issue by observing that the oracle query by the adversary can actually be removed.
Because the oracle is an encryption algorithm for single-bit messages and because the adversary queries to the oracle only polynomially many times,
we can remove the oracle by giving sufficiently many outputs of the oracle to the adversary in advance as an auxiliary input.
The duality in \cite{EC:HhaMorYam23} takes into account of the auxilially inputs to the adversary, and therefore now we can use the duality.

\paragraph{MACs with unclonable tags from UPSGs.}
It is straightforward to see that UPSGs imply EUF-CMA-secure MACs with quantum tags, because
we have only to take the output $\phi_k(x)$ of the UPSG on input $x$ as the tag corresponding to the message $x$.
However, the mere fact that the tags are quantum does not automatically mean that they are unclonable.
PRFSs also imply EUF-CMA-secure MACs with quantum tags, and in that case, the unclonability of tags is straightforward, because
quantum pseudorandomness implies unclonability~\cite{C:JiLiuSon18}.
However, in the case of UPSGs, it is not clear whether the quantum unpredictability is also sufficient for unclonability.

Our idea to construct unclonable tags is to use the unclonability of random BB84 states. (In other words, to use Wiesner money~\cite{Wiesner83}.)
Assume that a UPSG exists. Then, there exists an EUF-CMA-secure MAC. (Actually, in the following argument, any EUF-CMA-secure MACs even with classical tags are fine.)
Let $\tau_m$ be a tag corresponding to a message $m$. Then, if we set
$\tau_m'\coloneqq \tau_m\otimes |x\rangle\langle x|_\theta$ as a new tag,
it becomes unclonable.
Here, $x,\theta$ are random bit strings, $|x\rangle_\theta\coloneqq\bigotimes_iH^{\theta^i}|x^i\rangle$,
$H$ is the Hadamard gate, 
and $x^i$ and $\theta^i$ are $i$th bit of $x$ and $\theta$, respectively.

However, the verifier who wants to verify the tag cannot verify $\tau_m'$, 
because the verifier does not know $x$ and $\theta$.
Let us therefore modify our tag as
$\tau_m''\coloneqq (x,\theta,\tau_m\otimes|x\rangle\langle x|_\theta)$.
Now, this can be verified by doing the projection onto $|x\rangle_\theta$, but 
the unclonability is no longer satisfied because $x$ and $\theta$ are open.

To solve the issue, we introduce IND-CPA-secure SKE. Fortunately, as we show in this paper, IND-CPA-secure SKE exists if UPSGs exist.
Let us modify our tag as
$\tau_m'''\coloneqq \Enc(\sk,(x,\theta))
\otimes \tau_m\otimes|x\rangle\langle x|_\theta$,
where $\Enc$ is the encryption algorithm of the SKE scheme.
Now it is unclonable due to the security of the SKE scheme, but it is no longer authenticated:
$\Enc(\sk,(x,\theta))
\otimes \tau_m\otimes|x\rangle\langle x|_\theta$
could be replaced with
$\Enc(\sk,(x',\theta'))
\otimes \tau_m\otimes|x'\rangle\langle x'|_{\theta'}$
with another $x'$ and $\theta'$ chosen by the adversary,
because encryption does not necessarily mean authentication.
The adversary who knows $x'$ and $\theta'$ can of course make many copies of the tag.

The problem is finally solved by considering the following tag:
$\tau_m''''\coloneqq \Enc(\sk,\tau_{m\|x\|\theta}\otimes |x,\theta\rangle\langle x,\theta|)
\otimes |x\rangle\langle x|_\theta$,
where $\tau_{m\|x\|\theta}$ is the tag corresponding to the message $m\|x\|\theta$.

\subsection{Open Problems}
To conclude Introduction, let us provide some interesting open problems.
\begin{enumerate}
    \item 
    Do UPSGs imply PRFSs? Or can we separate them?
    \item 
    Is there any application that is possible with PRFSs, but not with UPSGs? So far, all known applications of PRFSs are achievable with UPSGs.
    \item 
    We show that EUF-CMA-secure MACs are possible with UPSGs.
   How about EUF-CMA-secure digital signatures? Can we realize them with UPSGs?
    So far, we do not know how to realize them even with PRUs.\footnote{Recently, \cite{coladangelo2024black} showed an oracle separation between PRUs and EUF-CMA-secure digital signatures with classical signatures.}
    \item 
    Do OWSGs imply UPSGs? It is neither known whether PRSGs imply PRFSs.
\end{enumerate} 

\begin{figure}
\begin{tikzcd}
&&
\text{OWFs} \arrow[rrdddd, "{\cite{TCC:AGQY22}}"]  
&&\\
&&&& \\
&&&&\\
&& 
\text{PRUs} \arrow[uuu,"{\cite{Kre21}\:\:\:\:\:}" cross] \arrow[rrd, "\cite{TCC:AGQY22}" description]
&&\\
&&&& 
\text{PRFSs} \arrow[lldd, "\text{\cref{thm:PRFS_to_UPSG}}" description,color=red] \arrow[lllld, "\text{Trivial}" description] \\
\text{PRSGs} \arrow[rrdddddddd, "{\cite{C:JiLiuSon18}}" description, bend right]
&&&&\\
&& 
\text{UPSGs} \arrow[ldd, "\text{Trivial}" description] \arrow[rdd, "\text{\cref{thm:SKE_from_UPSG}}" description,color=red]
&&\\
&&&&\\
& 
\text{EUF-CMA MACs} 
&& 
\text{IND-CPA SKE} \arrow[ld, no head,color=red] \arrow[lddddddd,"\cite{cryptoeprint:2022/1336}" description,bend left] \arrow[lddddddddd,"\cite{cryptoeprint:2022/1336}" description,bend left=49] 
&\\
&& 
{} \arrow[lu, no head,color=red] \arrow[dd, "\text{\cref{thm:unclonable_MAC_from_UPSG_and_SKE}}" description,color=red]
&&\\
&&&&\\
&& 
\text{unclonable MACs} \arrow[dd, "\text{Trivial}" description]
&&\\
&&&&\\
&&
\text{Private Money} \arrow[dd, "\cite{cryptoeprint:2022/1336}" description, dashed]
&&\\
&&&&\\
&&
\text{OWSGs} \arrow[dd, "\cite{cryptoeprint:2023/1620}" description, dashed]
&&\\
&&&&\\
&&
\text{EFIs}                                                              &&                                               
\end{tikzcd}
\caption{Relation among primitives.
{\color{red}The red color arrows} represent our results. A dotted arrow from primitive A to primitive B represents that primitive A with pure outputs implies primitive B.}
\label{fig:primitives}
\end{figure}
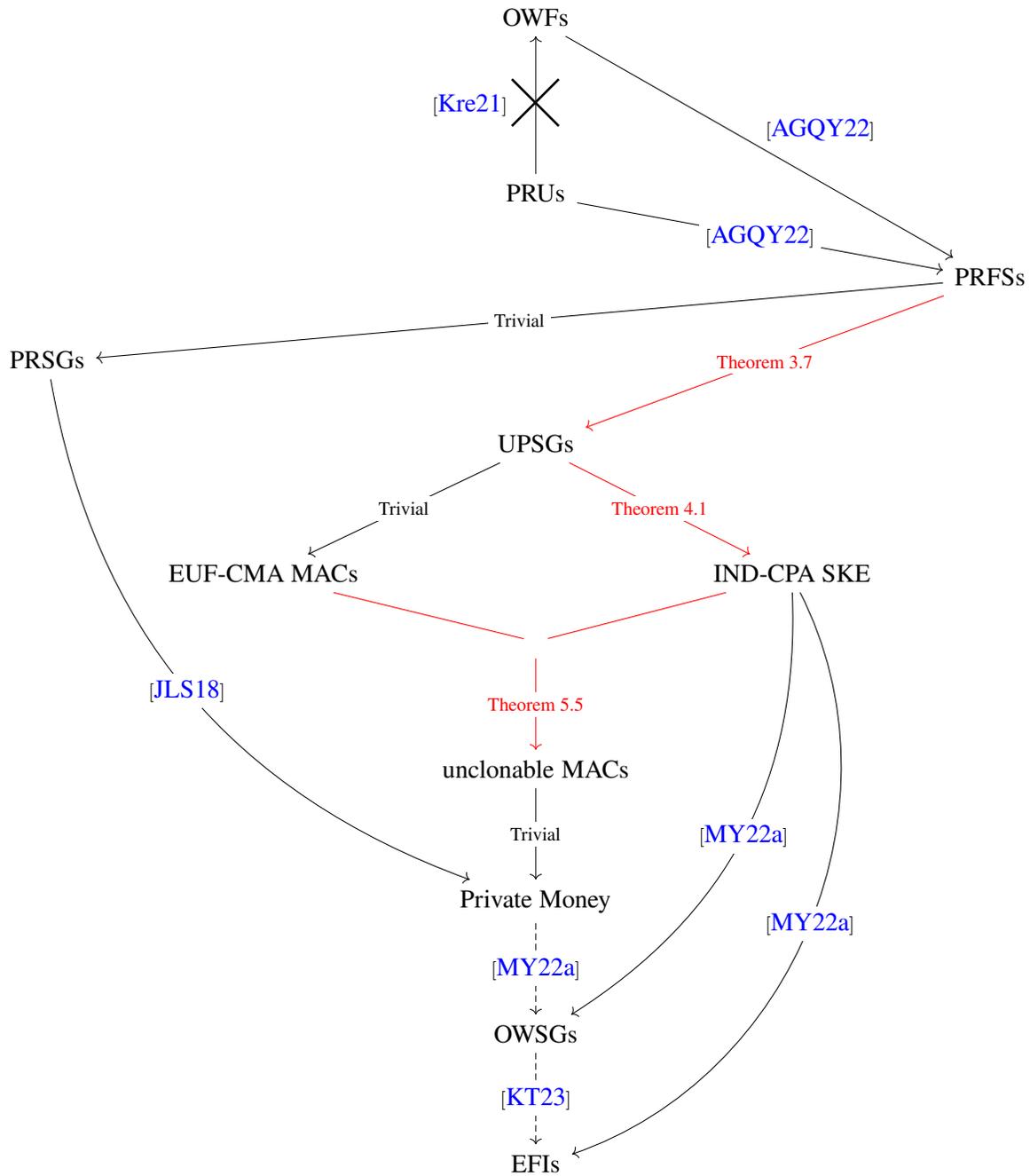

\clearpage
\section{Preliminaries}\label{sec:preliminaries}

\subsection{Basic Notations}
\label{sec:basic_notations}

We use the standard notations of quantum computing and cryptography. For a bit string $x$, $x^i$ denotes the $i$th bit of $x$.
For two bit strings $x$ and $y$, $x\|y$ means the concatenation of them.
We use $\secp$ as the security parameter.
$[n]$ means the set $\{1,2,...,n\}$.
For any set $S$, $x\gets S$ means that an element $x$ is sampled uniformly at random from the set $S$.
We write $\negl$ to mean a negligible function
and $\poly$ to mean a polynomial. 
PPT stands for (classical) probabilistic polynomial-time and QPT stands for quantum polynomial-time.
For an algorithm $A$, $y\gets A(x)$ means that the algorithm $A$ outputs $y$ on input $x$.

For simplicity, we sometimes omit the normalization factor of a quantum state.
(For example, we write $\frac{1}{\sqrt{2}}(|x_0\rangle+|x_1\rangle)$ just as
$|x_0\rangle+|x_1\rangle$.)
$I\coloneqq|0\rangle\langle0|+|1\rangle\langle 1|$ is the two-dimensional identity operator.
For the notational simplicity, we sometimes write $I^{\otimes n}$ just as $I$ when the dimension is clear from the context. 
We use $X$, $Y$ and $Z$ as Pauli operators.
For a bit string $x$, $X^x\coloneqq\bigotimes_iX^{x^i}$. We use $Y^y$ and $Z^z$ similarly.
For two density matrices $\rho$ and $\sigma$, the trace distance is defined as
${\rm TD}(\rho,\sigma)\coloneqq 
\frac{1}{2}\|\rho-\sigma\|_1 =
\frac{1}{2} \mathrm{Tr}\left[\sqrt{(\rho - \sigma)^2}\right]$,
where $\|\cdot\|_1$ is the trace norm.

\subsection{Lemmas}
We use the following lemma by Hhan, Morimae and Yamakawa~\cite{EC:HhaMorYam23} (based on~\cite{AAS20}).

\begin{lemma}[Duality Between Swapping and Distinction~\cite{EC:HhaMorYam23}, Theorem 5.1]
\label{lem:HMY22}
     Let $\ket{\psi}$ and $\ket{\phi}$ be orthogonal $n$-qubit states. Assume that a QPT algorithm $\cA$ with some $m$-qubit advice state $\ket{\tau}$ 
     can distinguish $\ket{\psi}$ and $\ket{\phi}$ with advantage $\Delta$. Then, there exists a polynomial-time implementable unitary $V$ over $(n+m)$-qubit states such that
        \begin{align}
            \frac{|\bra{\alpha}\bra{\tau}V\ket{\beta}\ket{\tau}+\bra{\beta}\bra{\tau}V\ket{\alpha}\ket{\tau}|}{2}=\Delta,
        \end{align}
        where $\ket{\alpha}\coloneqq\frac{\ket{\psi}+\ket{\phi}}{\sqrt{2}}$ and $\ket{\beta}\coloneqq\frac{\ket{\psi}-\ket{\phi}}{\sqrt{2}}$.
\end{lemma}

We also use the security of Wiesner money \cite{Wiesner83,MVW12_Wiesner_money}.
\begin{lemma}[Security of Wiesner Money~\cite{MVW12_Wiesner_money}]\label{lem:Wiesner_money}
   Let us consider the following security game:
   \begin{enumerate}
       \item The challenger $\cC$ chooses $x,\theta\gets\bit^\secp$ and sends $\ket{x}_{\theta}$ to the adversary $\cA$. Here, $\ket{x}_{\theta}\coloneqq\bigotimes_{i\in[\secp]}H^{\theta^i}\ket{x^i}$.
       \item $\cA$ sends a $2\secp$-qubit state $\rho$ to $\cC$.
       \item $\cC$ projects $\rho$ onto $\ket{x}^{\otimes 2}_{\theta}$. If the projection is successful, $\cC$ outputs $\top$. Otherwise, $\cC$ outputs $\bot$.
   \end{enumerate}
    For any unbounded adversary $\cA$, $\Pr[\top\gets\cC]\le\negl(\secp)$.
\end{lemma}

\subsection{Cryptographic Primitives}
The following is the standard definition of IND-CPA-secure SKE schemes for classical messages. However, in this paper, we consider general cases where
ciphertexts can be quantum states.
\begin{definition}[IND-CPA-Secure SKE for Classical Messages]\label{def:SKE_for_classical}
An IND-CPA-secure secret-key encryption (SKE) scheme for classical messages is a set of algorithms $(\KeyGen,\Enc,\Dec)$ such that
\begin{itemize}
    \item 
    $\KeyGen(1^\secp)\to \sk:$ It is a QPT algorithm that, on input the security parameter $\secp$,
    outputs a classical secret key $\sk$.
    \item 
    $\Enc(\sk,m)\to\ct:$
    It is a QPT algorithm that, on input $\sk$ and a classical bit string (plaintext) $m$, outputs a ciphertext $\ct$,
    which can be a quantum state in general.
    \item 
    $\Dec(\sk,\ct)\to m:$
    It is a QPT algorithm that, on input $\sk$ and $\ct$, outputs $m$.
\end{itemize}
We require the following two properties.

\paragraph{Correctness:}
For any bit string $m$,
\begin{align}
\Pr[m\gets\Dec(\sk,\ct):\sk\gets\KeyGen(1^\secp),\ct\gets\Enc(\sk,m)]\ge1-\negl(\secp).\label{eq:cSKE_correctness}
\end{align}

\paragraph{IND-CPA security (against classical query):}
For any QPT adversary $\cA$,
\begin{align}
\Pr\left[b=b':
\begin{array}{r}
\sk\gets\KeyGen(1^\secp)\\
(m_0,m_1,\st)\gets\cA^{\Enc(\sk,\cdot)}\\
b\gets\bit\\
\ct\gets\Enc(\sk,m_b)\\
b'\gets\cA^{\Enc(\sk,\cdot)}(\st,\ct)
\end{array}
\right]    
\le\frac{1}{2}+\negl(\secp),
\end{align}
where $\cA$ can only classically query $\Enc(\sk,\cdot)$.
\end{definition}

We also need IND-CPA-secure SKE for quantum messages. 
\begin{definition}[IND-CPA-Secure SKE for Quantum Messages~\cite{broadbent2015quantum,alagic2016computational}]\label{def:SKE_for_quantum}
An IND-CPA-secure secret-key encryption (SKE) scheme for quantum messages is a set of algorithms $(\KeyGen,\Enc,\Dec)$ such that
\begin{itemize}
    \item 
    $\KeyGen(1^\secp)\to \sk:$ It is a QPT algorithm that, on input the security parameter $\secp$,
    outputs a classical secret key $\sk$.
    \item 
    $\Enc(\sk,\rho)\to\ct:$
    It is a QPT algorithm that, on input $\sk$ and a quantum state $\rho$ on the register $\regM$, outputs a quantum state $\ct$ on the register $\regC$.
    \item 
    $\Dec(\sk,\ct)\to \rho:$
    It is a QPT algorithm that, on input $\sk$ and a state $\ct$ on the register $\regC$, outputs a state $\rho$ on the register $\regM$.
\end{itemize}
We require the following two properties.

\paragraph{Correctness:}
\begin{align}
\Exp_{\sk\gets\KeyGen(1^\secp)}
\|\Dec(\sk,\cdot)\circ\Enc(\sk,\cdot)-\mathrm{id}\|_{\diamond}\le\negl(\secp),
\end{align}
where $\mathrm{id}$ is the identity map,
$\Enc(\sk,\cdot)$ is a CPTP map\footnote{In this paper, we sometimes use the same notation $\Enc$ for an algorithm and a CPTP map, but we believe there is no confusion.} that runs the encryption algorithm $\Enc$ with $\sk$
on the plaintext state,
$\Dec(\sk,\cdot)$ is a CPTP map that runs the decryption algorithm $\Dec$ with $\sk$
on the ciphertext state,
and $\Dec(\sk,\cdot)\circ\Enc(\sk,\cdot)$ is the composition of $\Dec(\sk,\cdot)$ and $\Enc(\sk,\cdot)$. 
Here
$\|\cF-\cE\|_\diamond
\coloneqq\max_{\rho}\|(\cF\otimes \rm{id})(\rho)-(\cE\otimes\rm{id})(\rho)\|_1$ 
is the diamond norm between two CPTP maps $\cF$ and $\cE$ acting on $n$ qubits~\cite{watrous2018theory}, 
where the max is taken over all $2n$-qubit states $\rho$.

\paragraph{IND-CPA security:}
Let us consider the following security game:
\begin{enumerate}
    \item The challenger $\cC$ runs $\sk\gets\KeyGen(1^\secp)$.
    \item 
    The adversary $\cA$ can query the oracle $\Enc(\sk,\cdot)$. (This means that $\cA$ can apply the CPTP map $\Enc(\sk,\cdot)$
    on the register $\regM$ of any $\cA$'s state $\rho_{\regM,\regZ}$ over the registers $\regZ$ and $\regM$, 
    and get another state $\rho_{\regZ,\regC}'$ over the registers $\regZ$ and $\regC$.)
    \item $\cA$ sends two registers $\regM_0$ and $\regM_1$ to $\cC$. \label{step:challenge_query_send_def}
    \item $\cC$ chooses $b\gets\bit$ and applies the CPTP map $\Enc(\sk,\cdot)$ on $\regM_b$.
    $\cC$ then sends the output to $\cA$.\label{step:challenge_query_generate_def}
    \item
    $\cA$ can query the oracle $\Enc(\sk,\cdot)$.
    \item $\cA$ sends $b'\in\bit$ to $\cC$. 
    \item If $b=b'$, $\cC$ outputs $\top$. Otherwise, $\cC$ outputs $\bot$.
\end{enumerate}
For any QPT adversary $\cA$,
$
    \Pr[\top\gets\cC]\le\frac{1}{2}+\negl(\secp).
$
\end{definition}

The following lemma is essentially shown in \cite{broadbent2015quantum}. We give its proof in \cref{sec:proof_of_cSKE_to_qSKE}.

\begin{lemma}[IND-CPA security for classical messages implies that for quantum messages~\cite{broadbent2015quantum}]\label{lem:SKE_for_classical_imply_SKE_for_quantum}
If IND-CPA-secure SKE schemes for classical messages that are secure against QPT adversaries that query the encryption oracle classically
exist, then IND-CPA-secure SKE schemes for quantum messages exist.
\end{lemma}

The following lemma can be shown with the standard hybrid argument~\cite{broadbent2015quantum}.
\begin{lemma}[IND-CPA-multi security~\cite{broadbent2015quantum}]\label{lem:CPA_multi}
    Let $(\KeyGen,\Enc,\Dec)$ be an IND-CPA-secure SKE scheme for quantum messages. 
    Let $t$ be a polynomial.
    Let us consider the security game that is the same as that of \cref{def:SKE_for_quantum} except for the following two modifications.
    \begin{itemize}
        \item In step \ref{step:challenge_query_send_def}, $\cA$ sends two registers $\regM'_0$ and $\regM'_1$ to $\cC$.
        Here, $\regM'_0$ consists of $t$ registers $\{\regM^i_0\}_{i\in[t]}$, and
        $\regM'_1$ consists of $t$ registers $\{\regM^i_1\}_{i\in[t]}$. 
        For each $i\in[t]$ and $b\in\bit$, $|\regM^i_b|=|\regM_b|$, where $|\regA|$ is the size (i.e., the number of qubits) of the register $\regA$.
        \item In step \ref{step:challenge_query_generate_def}, $\cC$ chooses $b\gets\bit$ and applies the CPTP map $\Enc(\sk,\cdot)$ on each $\regM^i_b$ for $i\in[t]$. $\cC$ then sends the all outputs to $\cA$.
    \end{itemize}
    Then, in this modified game,
    $
        \Pr[\top\gets\cC]\le\frac{1}{2}+\negl(\secp)
    $
    for any QPT adversary $\cA$ and any polynomial $t$.
\end{lemma}

\begin{definition}[One-way States Generators (OWSGs)~\cite{cryptoeprint:2022/1336,C:MorYam22}]
    A one-way states generator (OWSG) is a set of algorithms $(\KeyGen,\StateGen,\Ver)$ such that
    \begin{itemize}
        \item $\KeyGen(1^\secp)\to k:$ It is a QPT algorithm that, on input the security parameter $\secp$, outputs a classical key $k$.
        \item $\StateGen(k)\to\phi_k:$ It is a QPT algorithm that, on input $k$, outputs a quantum state $\phi_k$.
        \item $\Ver(k',\phi_k)\to\top/\bot:$ It is a QPT algorithm that, on input $\phi_k$ and a bit string $k'$, outputs $\top$ or $\bot$.
    \end{itemize}
    We require the following correctness and security.
    
    \paragraph{Correctness:}
    \begin{align}
        \Pr[\top\gets\Ver(k,\phi_k):k\gets\KeyGen(1^\secp),\phi_k\gets\StateGen(k)]\ge1-\negl(\secp).
    \end{align}

    \paragraph{Security:}
    For any QPT adversary $\cA$ and any polynomial $t$,
    \begin{align}
        \Pr[\top\gets\Ver(k',\phi_k):k\gets\KeyGen(1^\secp),\phi^{\otimes t}_k\gets\StateGen(k)^{\otimes t},k'\gets\cA(1^\secp,\phi^{\otimes t}_k)]\le\negl(\secp).
    \end{align}
    Here, $\phi^{\otimes t}_k\gets\StateGen(k)^{\otimes t}$
    means that the $\StateGen$ algorithm is run $t$ times.
\end{definition}

\section{Unpredictable State Generators}
\label{sec:UPSG}

\subsection{Definition}
\label{sec:definition_UPSG}
In this subsection, we define UPSGs.
The syntax is given as follows.

\begin{definition}[Unpredictable States Generators (UPSGs)]
\label{def:UPSGs}
An unpredictable states generator is a set $(\KeyGen,\Eval)$ of QPT algorithms such that
\begin{itemize}
    \item 
    $\KeyGen(1^\secp)\to k:$
    It is a QPT algorithm that, on input the security parameter $\secp$, outputs a classical key $k$.
    \item 
    $\Eval(k,x)\to (x,\phi_k(x)):$
    It is a QPT algorithm that, on input $k$ and a bit string $x$, outputs $x$ and a quantum state $\phi_k(x)$.
\end{itemize}
\end{definition}

In general, $\phi_k(x)$ could be mixed states, but in this paper, we restrict them to pure states.

The security, which we call unpredictability, roughly means that
no QPT adversary (who can {\it quantumly} query to $\Eval(k,\cdot)$) can output $(x^*,\rho)$ such that $x^*$ was {\it not queried} and $\rho$ is close to $|\phi_k(x^*)\rangle$.
In order to formally define it, we have to clarify what we mean by ``quantumly query'' and ``not queried before''.

\paragraph{Quantum query.}
We assume that $|\phi_k(x)\rangle\gets\Eval(k,x)$ is the following QPT algorithm:
on input $k$ and $x$, it applies a unitary $U_k$ on $\ket{x}_{\regX}\ket{0...0}_{\regY,\regZ}$ 
to generate $\ket{x}_{\regX}\ket{\phi_k(x)}_{\regY}\ket{{\rm junk}_k}_{\regZ}$ and outputs the $\regX$ and $\regY$ registers.
Note that it is not the most general case. 
First, as we have mentioned, 
we assume that the output $|\phi_k(x)\rangle$ is pure. 
Second, in general, the junk state $\ket{{\rm junk}_k}$ could depend on $x$,
but we here assume that it depends only $k$. These two restrictions seem to be necessary to
well define the quantum query.

With such $\Eval$, the quantum query to the oracle $\Eval(k,\cdot)$ means the following:
\begin{enumerate}
    \item A state $\sum_x \alpha_x \ket{x}_{\regX}\ket{\xi_x}$ is input to the oracle, where $\{\alpha_x\}_x$ are any complex coefficients and $\{\ket{\xi_x}\}_x$ are any states.
    \item The oracle adds the ancilla state $\ket{0...0}_{\regY,\regZ}$ and applies $U_k$ on the registers $\regX,\regY,\regZ$ 
    of $\sum_x \alpha_x \ket{x}_{\regX}\ket{0...0}_{\regY,\regZ}\ket{\xi_x}$ to generate $\sum_x \alpha_x \ket{x}_{\regX}\ket{\phi_k(x)}_{\regY}\ket{{\rm junk}_k}_{\regZ}\ket{\xi_x}$.
    \item The oracle removes the junk register $\regZ$ and outputs the state $\sum_x\alpha_x\ket{x}_{\regX}\ket{\phi_k(x)}_{\regY}\ket{\xi_x}$.
\end{enumerate}

\paragraph{Not queried.}
We define the word ``not queried'' as follows.
Assume that $\cA$ queries the oracle $q$ times.
For each $i\in[q]$, let $\ket{\psi_i}$ be the {\it entire} $\cA$'s state immediately before its $i$th query to the oracle. 
(Without loss of generality, we can assume that $\cA$ postpones all measurements to the last step, and then $\cA$'s entire state is
always pure.)
We say that $x^*$ is not queried if 
$\langle\psi_i|(|x^*\rangle\langle x^*|_{\regX}\otimes I)|\psi_i\rangle=0$ for all $i\in[q]$.
Here, for each $i\in[q]$, $\ket{\psi_i}=\sum_x c_x |x\rangle_\regX\otimes|\eta_x\rangle$.

Now we define the unpredictability.
\begin{definition}[Unpredictability]\label{def:unpredictability}
   Let us consider the following security game:
   \begin{enumerate}
       \item 
       The challenger $\cC$ runs $k\gets\KeyGen(1^\secp)$.
       \item 
       The adversary $\cA^{\Eval(k,\cdot)}(1^\secp)$ outputs a bit string $x^*$ and a quantum state $\rho$, and sends them to $\cC$. 
       Here, $\cA$ can make quantum queries to $\Eval(k,\cdot)$. $x^*$ should not be queried by $\cA$.
       \item 
      $\cC$ projects $\rho$ onto $\ket{\phi_k(x^*)}$.
      If the projection is successful, $\cC$ outputs $\top$.
      Otherwise, $\cC$ outputs $\bot$.
   \end{enumerate}
   For any QPT adversary $\cA$, $\Pr[\top\gets\cC]\le\negl(\secp)$.
\end{definition}

\begin{remark}\label{rem:UPSG}
Note that the projection of $\rho$ onto $\ket{\phi_k(x^*)}$ can be done as follows:
\begin{enumerate}
    \item 
    Prepare $|x^*\rangle\langle x^*|\otimes\rho\otimes|{\rm junk}_k\rangle\langle{\rm junk}_k|$.
    \item 
    Apply $U^{\dag}_{k}$ on $\ket{x^*}\bra{x^*}\otimes\rho\otimes \ket{{\rm junk}_{k}}\bra{{\rm junk}_{k}}$.
    \item 
    Measure all qubits in the computational basis. If the result is $x^*\|0...0$, the projection is successful.
    Otherwise, the projection is failed.
\end{enumerate}
\end{remark}

\begin{remark}
It is easy to see that    
UPSGs with $O(\log\secp)$-qubit output do not exist.\footnote{The adversary has only to output $0...0$ and maximally-mixed state.}
\end{remark}

\begin{remark}
   In \cite{C:BonZha13}, they define a security of digital signatures against quantum adversaries. Their security definition is as follows: any QPT quantum adversary, who queries the signing oracle $t$ times, cannot output $t+1$ valid message-signature pairs. 
   We could define a quantum version of unpredictability based on their security definition, but exploring this possibility is beyond the scope of the present paper. 
   At least, their definition seems to be incomparable to \cref{def:unpredictability}.
   In particular, we do not know how to construct IND-CPA-secure SKE from their definition, because
   we do not know how to use the duality in that case.
\end{remark}

\subsection{Relation to PRFSs}
In this section, we recall the definition of PRFSs and construct UPSGs from PRFSs.

\begin{definition}[Pseudorandom Function-Like States (PRFSs)~\cite{C:AnaQiaYue22,TCC:AGQY22}]
A pseudorandom function-like state (PRFS) (generator) is a set of algorithms $(\KeyGen,\Eval)$ such that
\begin{itemize}
    \item 
    $\KeyGen(1^\secp)\to k:$ It is a QPT algorithm that, on input the security parameter $\secp$, outputs a classical secret key $k$.
   \item 
   $\Eval(k,x)\to \ket{\phi_k(x)}:$ It is a QPT algorithm that on input $k$ and a bit string $x$, outputs a quantum state $\ket{\phi_k(x)}$.
\end{itemize}
We require the following security.
For any QPT adversary $\cA$,
\begin{align}
|\Pr[1\gets\cA^{\Eval(k,\cdot)}(1^\secp)]-\Pr[1\gets\cA^{\mathcal{O}_{Haar}}(1^\secp)]|\le \negl(\secp).
\end{align}
Here, $\cA^{\Eval(k,\cdot)}$ means that $\cA$ can quantumly query the oracle $\Eval(k,\cdot)$ in the sense of \cref{sec:definition_UPSG}\footnote{In \cite{C:AnaQiaYue22,TCC:AGQY22}, they do not explicitly consider the junk state
$\ket{{\rm junk}_k}$. Here, we assume that $\ket{{\rm junk}_k}$ is independent of $x$ similarly to the case of UPSGs.}.
$\cA^{\mathcal{O}_{Haar}}$ means that $\cA$ can quantumly query the oracle $\mathcal{O}_{Haar}$ in the following sense.
\begin{enumerate}
  \item A state $\sum_x \alpha_x \ket{x}_{\regX}\ket{\xi_x}$ is input to the oracle, where $\{\alpha_x\}_x$ are any complex coefficients and $\{\ket{\xi_x}\}_x$ are any states.
    \item The oracle returns $\sum_x\alpha_x\ket{x}_{\regX}\ket{\psi_x}_{\regY}\ket{\xi_x}$,
    where $\ket{\psi_x}$ is a Haar random state.
\end{enumerate}

\end{definition}

\begin{theorem}\label{thm:PRFS_to_UPSG}
If PRFSs exist then UPSGs exist.    
\end{theorem}

\begin{proof}[Proof of \cref{thm:PRFS_to_UPSG}]
Let $(\KeyGen,\Eval)$ be a PRFS. We show that it is a UPSG.
Assume that it does not satisfy the unpredictability.
Then, there exist a polynomial $p$ and
a QPT adversary $\cA$ that can quantumly query $\Eval(k,\cdot)$ such that 
\begin{align}
\sum_{k}\Pr[k\gets\KeyGen(1^\secp)]
\sum_{x^*}\langle x^*|\langle \phi_k(x^*)|
\cA^{\Eval(k,\cdot)}(1^\secp)|x^*\rangle|\phi_k(x^*)\rangle 
\ge\frac{1}{p(\secp)}\label{eq:assume_UPSG_not_secure}
\end{align}

for infinitely many $\secp\in\mathbb{N}$. Here, $\cA^{(\cdot)}(1^\secp)$ denotes the state of $\cA^{(\cdot)}$ before the measurement.
Then, the following QPT adversary $\cB$ breaks the security of PRFS.
\begin{enumerate}
\item 
The challenger $\cC'$ of the PRFS chooses $b\gets\bit$.
    \item Run $\cA$ on input $1^\secp$. When $\cA$ queries the oracle, $\cB$ simulates it by querying $\cB$'s oracle (that is $\Eval(k,\cdot)$ if $b=0$ and $\mathcal{O}_{Haar}$ if $b=1$).
    \item $\cB$ measures the first register of $\cA^{(\cdot)}(1^\secp)$ to get $x^*$.
    Query $x^*$ to $\cB$'s oracle to get $\ket{\xi}$, 
    which is $\ket{\xi}=\ket{\phi_k(x^*)}$ if $b=0$ and a Haar random state $\ket{\psi_{x^*}}$ if $b=1$.
    \item $\cB$ does the swap test between the second register of $\cA^{(\cdot)}(1^\secp)$ and $\ket{\xi}$. 
    If the swap test succeeds, $\cB$ outputs 1. 
    Otherwise, $\cB$ outputs $0$.
\end{enumerate}
If $b=0$,
\begin{align}
    \Pr[1\gets\cB]&=\frac{1}{2}+\frac{1}{2}\sum_{k}\Pr[k\gets\KeyGen(1^\secp)]\sum_{x^*}\langle x^*|\langle\phi_k(x^*)|\cA^{\Eval(k,\cdot)}(1^\secp)|x^*\rangle|\phi_k(x^*)\rangle\\
    &\ge\frac{1}{2}+\frac{1}{2p(\secp)}
\end{align}
for infinitely many $\secp$. Here we have used \cref{eq:assume_UPSG_not_secure}.
On the other hand, if $b=1$,
\begin{align}
    \Pr[1\gets\cB]&=\frac{1}{2}+\frac{1}{2}\sum_{k}\Pr[k\gets\KeyGen(1^\secp)]\sum_{x^*}
    \Exp_{|\psi\rangle\gets\mu}\langle x^*|\langle\psi|\cA^{\mathcal{O}_{Haar}}(1^\secp)|x^*\rangle|\psi\rangle\\
    &\le\frac{1}{2}+\negl(\secp),
\end{align}
where $\mu$ denotes the Haar measure and we have used $\Exp_{\ket{\psi}\gets\mu}\bra{\psi}\sigma\ket{\psi}\le\negl(\secp)$ for any state $\sigma$. Therefore, $\cB$ breaks the security of the PRFS.
\end{proof}

\subsection{Pure OWSGs from UPSGs}
In this section, we show that UPSGs imply OWSGs with pure output states.
\begin{theorem}
If UPSGs exist, then pure OWSGs exist. 
\end{theorem}

\begin{proof}
Let $(\mathsf{UPSG}.\KeyGen,\mathsf{UPSG}.\Eval)$ be a UPSG. 
From it, we construct a pure OWSG $(\KeyGen,\StateGen)$ as follows.
\begin{itemize}
    \item 
    $\KeyGen(1^\secp)\to k':$ Run $k\gets\mathsf{UPSG}.\KeyGen(1^\secp)$. 
    Choose $x_i\gets\bit^\ell$ for $i\in[n]$.
    Here, $n\coloneqq |k|+\secp$.
    Output $k'\coloneqq(k,x_1,...,x_n)$.
    \item 
    $\StateGen(k')\to\psi_{k'}:$
    Parse $k'=(k,x_1,...,x_n)$.
    Run $\phi_{k}(x_i)\gets \mathsf{UPSG}.\Eval(k,x_i)$ for $i\in[n]$.
Output $\psi_{k'}\coloneqq
\left(\bigotimes_{i=1}^n\phi_{k}(x_i)\right)\otimes(\bigotimes_{i=1}^n\ket{x_i})$.
\end{itemize}
For the sake of contradiction, assume that this construction is not secure.
This means that there exist polynomials $p$ and $t$, and a QPT adversary $\cA$ such that
\begin{align}
\frac{1}{p(\secp)}
&\le
\sum_{k}\Pr[k]
\frac{1}{2^{n\ell}}\sum_{x_1,...,x_n}
\sum_{s,x_1',...,x_n'}{\rm Pr'}[s,x_1',...,x_n'|k,x_1,...,x_n]
\prod_{i\in[n]}|\langle\phi_k(x_i)|\phi_s(x_i')\rangle|^2\delta_{x_i,x_i'}\\
&=
\sum_{k}\Pr[k]
\frac{1}{2^{n\ell}}\sum_{x_1,...,x_n}
\sum_{s}{\rm Pr'}[s,x_1,...,x_n|k,x_1,...,x_n]
\prod_{i\in[n]}|\langle\phi_k(x_i)|\phi_s(x_i)\rangle|^2
\end{align}
for infinitely many $\secp$.
Here, $\Pr[k]\coloneqq\Pr[k\gets\mathsf{UPSG}.\KeyGen(1^\secp)]$ and
\begin{align}
{\rm Pr}'[s,x_1',...,x_n'|k,x_1,...,x_n]\coloneqq
\Pr[(s,x_1',...,x_n')\gets\cA(1^\secp,((\bigotimes_{i}\phi_k(x_i))\otimes(\bigotimes_{i}|x_i\rangle))^{\otimes t})].
\end{align} 
Define 
\begin{align}
K\coloneqq \left\{
k:
\frac{1}{2^{n\ell}}\sum_{x_1,...,x_n}
\sum_{s}{\rm Pr'}[s,x_1,...,x_n|k,x_1,...,x_n]
\prod_{i\in[n]}|\langle\phi_k(x_i)|\phi_s(x_i)\rangle|^2
\ge\frac{1}{2p(\secp)}
\right\}.
\end{align}
Then, from the standard average argument, 
\begin{align}
\sum_{k\in K}\Pr[k]\ge\frac{1}{2p(\secp)}
\label{K}
\end{align}
for infinitely many $\secp$.
Define
\begin{align}
X_k\coloneqq\left\{(x_1,...,x_n):
\sum_{s}{\rm Pr'}[s,x_1,...,x_n|k,x_1,...,x_n]
\prod_{i\in[n]}|\langle\phi_k(x_i)|\phi_s(x_i)\rangle|^2
\ge\frac{1}{4p(\secp)}
\right\}.
\end{align}
Then, from the standard average argument, 
for any $k\in K$,
\begin{align}
\frac{1}{2^{n\ell}}\sum_{(x_1,..,x_n)\in X_k}\ge\frac{1}{4p(\secp)}
\label{X}
\end{align}
for infinitely many $\secp$.
Define 
\begin{align}
N_{k}\coloneqq
\left\{s: \Pr_{x\gets\bit^\ell}
\left[
|\langle\phi_{s}(x)|\phi_k(x)\rangle|^2
\ge \frac{1}{8p(\secp)}\right]\ge \frac{1}{2}\right\}. 
\end{align}
For any $k$ and $s\notin N_{k}$, we have 
\begin{align}
\Pr_{x_1,...,x_n\gets\bit^\ell}
\left[\prod_{i\in[n]}|\langle\phi_{s}(x_i)|\phi_k(x_i)\rangle|^2\ge  \frac{1}{8p(\secp)}\right]\le 2^{-n}. 
\end{align}
This is because to satisfy $\prod_{i}|\langle\phi_{s}(x_i)|\phi_k(x_i)\rangle|^2\ge  1/8p(\secp)$, we must have  $|\langle\phi_{s}(x_i)|\phi_k(x_i)\rangle|^2\ge  1/8p(\secp)$ for all $i$, and the probability that it holds is at most $1/2$ for each $i$ by the assumption that $s\notin N_{k}$.
By the union bound, for any $k$,
\begin{align}
&\Pr_{x_1,...,x_n\gets\bit^\ell}\left[
\forall s\in \bit^{|k|}\setminus N_{k},
\prod_{i\in[n]}|\langle\phi_{s}(x_i)|\phi_k(x_i)\rangle|^2\le  \frac{1}{8p(\secp)}\right]\\
&\ge 1-(2^{|k|}-|N_k|)\cdot 2^{-n}\\
&\ge1- 2^{-n+|k|}.
\label{unionY}
\end{align}
Define
\begin{align}
Y_k\coloneqq\left\{
(x_1,...,x_n):
\forall s\in \bit^{|k|}\setminus N_{k},
\prod_{i\in[n]}|\langle\phi_{s}(x_i)|\phi_k(x_i)\rangle|^2\le  \frac{1}{8p(\secp)}
\right\}.
\end{align}
Then, \cref{unionY} means
\begin{align}
\frac{1}{2^{n\ell}}\sum_{(x_1,...,x_n)\in Y_k}\ge1-2^{-n+|k|}
\label{Y}
\end{align}
for all $k$. 
From the union bound, \Cref{X}, and \cref{Y}, 
for any $k\in K$,
\begin{align}
\frac{1}{2^{n\ell}}\sum_{(x_1,...,x_n)\in X_k\cap Y_k}\ge\frac{1}{4p(\secp)}-2^{-n+|k|}
\label{XY}
\end{align}
for infinitely many $\secp$.
Then for any $k\in K$ and any $(x_1,...,x_n)\in X_k\cap Y_k$,
\begin{align}
\frac{1}{4p(\secp)}
&\le
\sum_{s\in N_k}{\rm Pr'}[s,x_1,...,x_n|k,x_1,...,x_n]
\prod_{i\in[n]}|\langle\phi_k(x_i)|\phi_s(x_i)\rangle|^2\\
&+\sum_{s\not\in N_k}{\rm Pr'}[s,x_1,...,x_n|k,x_1,...,x_n]
\prod_{i\in[n]}|\langle\phi_k(x_i)|\phi_s(x_i)\rangle|^2\\
&\le
\sum_{s\in N_k}{\rm Pr'}[s,x_1,...,x_n|k,x_1,...,x_n]
\prod_{i\in[n]}|\langle\phi_k(x_i)|\phi_s(x_i)\rangle|^2
+\frac{1}{8p(\secp)}\\
&\le
\sum_{s\in N_k}{\rm Pr'}[s,x_1,...,x_n|k,x_1,...,x_n]
+\frac{1}{8p(\secp)},
\end{align}
which gives
\begin{align}
\sum_{s\in N_k}{\rm Pr'}[s,x_1,...,x_n|k,x_1,...,x_n]
\ge\frac{1}{8p(\secp)}
\label{s_A}
\end{align}
for any $k\in K$ and any $(x_1,...,x_n)\in X_k\cap Y_k$.

From the $\cA$, we construct a QPT adversary $\cB$ that breaks the security of the UPSG as follows.
\begin{enumerate}
\item Sample $x_1,...,x_n\gets\bit^\ell$ and $x^*\gets\bit^\ell$.
\item For each $i\in[n]$, query $x_i$ to the oracle $\mathsf{UPSG}.\Eval(k,\cdot)$ $t$ times to get $\phi_k(x_i)^{\otimes t}$.
\item Run $(s,x_1',...,x_n')\gets\A(1^\secp,(\bigotimes_{i=1}^n\phi_k(x_i))\otimes (\bigotimes_{i=1}^n \ket{x_i}))^{\otimes t})$.
If $x_i'\neq x_i$ for at least one $i\in[n]$, abort.
\item Query $x^*$ to the oracle $\mathsf{UPSG}.\Eval(k,\cdot)$ to get $\phi_k(x^*)$. Output $(x^*,\phi_k(x^*))$. 
\end{enumerate}
The probability that $\cB$ wins is
\begin{align}
&\frac{1}{2^\ell}\sum_{x^*}
\sum_{k}\Pr[k]
\frac{1}{2^{n\ell}}\sum_{x_1,...,x_n}
\sum_{s}{\rm Pr'}[s,x_1,...,x_n|k,x_1,...,x_n]
|\langle\phi_k(x^*)|\phi_s(x^*)\rangle|^2\\
&\ge
\sum_{k\in K}\Pr[k]
\frac{1}{2^{n\ell}}\sum_{(x_1,...,x_n)\in X_k\cap Y_k}
\sum_{s\in N_k}{\rm Pr'}[s,x_1,...,x_n|k,x_1,...,x_n]
\frac{1}{2^\ell}\sum_{x^*}
|\langle\phi_k(x^*)|\phi_s(x^*)\rangle|^2\\
&\ge
\frac{1}{2p(\secp)}\left(\frac{1}{4p(\secp)}-2^{-n+|k|}\right)\frac{1}{8p(\secp)}\frac{1}{8p(\secp)}\frac{1}{2}
\ge
\frac{1}{\poly(\secp)}
\end{align}
for infinitely many $\secp$, which means that $\cB$ breaks the security of the OWSGs.
Here, we have used the definition of $N_k$, \cref{s_A},
\cref{XY}, and \cref{K}.
\end{proof}
\section{IND-CPA-Secure SKE from UPSGs}
\label{sec:SKE_from_UPSG}
In this section, we construct IND-CPA secure SKE from UPSGs.

\begin{theorem}\label{thm:SKE_from_UPSG}
If UPSGs exist, then IND-CPA-secure SKE schemes  
for classical messages secure against classically querying QPT adversaries exist.    
\end{theorem}

\begin{remark}
From \cref{lem:SKE_for_classical_imply_SKE_for_quantum}, IND-CPA-secure SKE schemes for classical messages secure against classically querying
QPT adversaries imply IND-CPA-secure SKE schemes for quantum messages.
Therefore, the above theorem also shows the existence of such SKE schemes if UPSGs exist.
\end{remark}

\begin{proof}[Proof of \cref{thm:SKE_from_UPSG}]
It suffices to construct an IND-CPA-secure SKE scheme for single-bit messages because, from it, we can construct an IND-CPA-secure SKE scheme for multi-bit messages by parallel repetition.\footnote{See \cite{books/crc/KatzLindell2007}.} Let $(\mathsf{UPSG}.\KeyGen,\mathsf{UPSG}.\Eval)$ be a UPSG.
As is explained in \cref{sec:definition_UPSG}, we assume that $\mathsf{UPSG}.\Eval$ is the following algorithm: on input $k$ and $x\in\bit^\ell$, it applies a unitary 
$U_k$ on $\ket{x}_{\regX}\ket{0...0}_{\regY,\regZ}$ to generate $\ket{x}_{\regX}\ket{\phi_k(x)}_{\regY}\ket{{\rm junk}_k}_{\regZ}$, and outputs the registers $\regX$ and $\regY$.
From $(\mathsf{UPSG}.\KeyGen,\mathsf{UPSG}.\Eval)$, we construct an IND-CPA-secure SKE scheme $(\KeyGen,\Enc,\Dec)$ for single-bit messages as follows.
\begin{itemize}
    \item 
    $\KeyGen(1^\secp)\to \sk:$
    Run $k\gets\mathsf{UPSG}.\KeyGen(1^\secp)$ and output $\sk\coloneqq k$.
    \item 
    $\Enc(\sk,b)\to\ct:$
    Parse $\sk=k$.
    Choose $x,y\gets\bit^\ell$. Generate 
    \begin{align}
    \ket{\ct^{b}_{x,y}}_{\regX,\regY}\coloneqq\frac{\ket{0\|x}_{\regX}\ket{\phi_k(0\|x)}_{\regY}+(-1)^b\ket{1\|y}_{\regX}\ket{\phi_k(1\|y)}_{\regY}}{\sqrt{2}}
    \end{align}
    and output $\ct\coloneqq(x,y,\ket{\ct^{b}_{x,y}})$.
    Here, $\ket{\ct^b_{x,y}}$ is generated as follows:
    \begin{enumerate}
        \item 
        Prepare $\ket{0\|x}_{\regX}\ket{0...0}_{\regY,\regZ}+(-1)^b\ket{1\|y}_{\regX}\ket{0...0}_{\regY,\regZ}$.
        \item 
        Apply $U_k$ on the registers $\regX$, $\regY$, and $\regZ$ to generate
        \begin{align}
        \ket{0\|x}_{\regX}\ket{\phi_k(0\|x)}_{\regY}\ket{{\rm junk}_k}_{\regZ}+(-1)^b\ket{1\|y}_{\regX}\ket{\phi_k(1\|y)}_{\regY}\ket{{\rm junk}_k}_{\regZ}.
        \end{align}
        \item 
        Remove the register $\regZ$.
    \end{enumerate}
    \item 
    $\Dec(\sk,\ct)\to b':$ 
    Parse $\sk=k$ and $\ct=(x,y,\rho_{\regX,\regY})$. Run the following algorithm.
    \begin{enumerate}
        \item Prepare $\rho_{\regX,\regY}\otimes\ket{{\rm junk}_k}\bra{{\rm junk}_k}_{\regZ}$.
        \item Apply $U^{\dag}_k$ on  $\rho_{\regX,\regY}\otimes\ket{{\rm junk}_k}\bra{{\rm junk}_k}_{\regZ}$.
        \item Apply $\ket{0}\bra{0}\otimes X^x+\ket{1}\bra{1}\otimes X^y$ on the register $\regX$. 
        \item Measure the first qubit of the register $\regX$ in the Hadamard basis to get $b'\in\bit$. Output $b'$.
    \end{enumerate}
\end{itemize}
Correctness is clear. To show the security, we define Hybrid 0, which is the original security game of the IND-CPA-secure SKE scheme
between the challenger $\cC$ and the QPT adversary $\cA$, as follows.

\paragraph{Hybrid 0}

\begin{enumerate}
     \item The challenger $\cC$ runs $k\gets\mathsf{UPSG.}\KeyGen(1^{\lambda})$.
    \item $\cC$ chooses $b\gets\bit$ and $x,y\gets\bit^\ell$. $\cC$ generates $\ket{\ct_{x,y}^b}$ by running $\mathsf{UPSG.}\Eval(k,\cdot)$ coherently. Here,
    \begin{align}
        \ket{\ct^{b}_{x,y}}=\frac{\ket{0\|x}\ket{\phi_k(0\|x)}+(-1)^b\ket{1\|y}\ket{\phi_k(1\|y)}}{\sqrt{2}}.
    \end{align}
    \item $\cC$ sends $\ct\coloneqq(x,y,\ket{\ct_{x,y}^b})$ to the adversary $\cA$.
    \item 
    $\cA$ can classically query to the oracle $\mathcal{O}_k$, where $\mathcal{O}_k$ works as follows:
    \begin{enumerate}
        \item On input $c\in\bit$, it chooses $x',y'\gets\bit^\ell$ and generates $\ket{\ct_{x',y'}^c}$.
        \item It outputs $(x',y',\ket{\ct_{x',y'}^c})$.
    \end{enumerate} 
    $\cA$ sends $b'\in\bit$ to $\cC$.
    \item If $b=b'$, $\cC$ outputs $\top$. Otherwise, $\cC$ outputs $\bot$.
\end{enumerate}

For the sake of contradiction, assume that our construction is not IND-CPA secure. This means that there exist a polynomial $p$ and a QPT adversary $\cA$ such that
\begin{align}
     \Pr[\top\gets{\rm Hybrid~0}]\ge\frac{1}{2}+\frac{1}{p(\secp)}\label{assumption for SKE}
\end{align}
for infinitely-many $\secp\in\mathbb{N}$. 

Our goal is to construct a QPT adversary $\cB$ that breaks the unpredictability of the UPSG. For that goal,
we use the duality between swapping and distinction~\cite{EC:HhaMorYam23}. However, we cannot directly use it here,
because our $\cA$ queries to the encryption oracle $\mathcal{O}_k$, but the distinguisher in \cref{lem:HMY22} does not access any oracle.
To solve the issue,
we have to remove the oracle $\mathcal{O}_k$ from Hybrid 0. Fortunately, $\mathcal{O}_k$ is an encryption oracle for single-bit messages, and
$\cA$ makes classical queries only polynomial times.
Therefore, we can give $\cA$ enough number of outputs of $\mathcal{O}_k$ in advance as auxiliary inputs, and $\cA$ can use these states
instead of the outputs of $\mathcal{O}_k$. In this way, we can remove the oracle $\mathcal{O}_k$.
We formalize this as Hybrid 1.
It is clear that $\Pr[\top\gets{\rm Hybrid~1}]=\Pr[\top\gets{\rm Hybrid~0}]$.

\paragraph{Hybrid 1}

\begin{enumerate}
     \item The challenger $\cC$ runs $k\gets\mathsf{UPSG.}\KeyGen(1^{\lambda})$.
    \item $\cC$ chooses $b\gets\bit$ and $x,y\gets\bit^\ell$. $\cC$ generates $\ket{\ct_{x,y}^b}$ by running $\mathsf{UPSG.}\Eval(k,\cdot)$ coherently. Here,
    \begin{align}
        \ket{\ct^{b}_{x,y}}=\frac{\ket{0\|x}\ket{\phi_k(0\|x)}+(-1)^b\ket{1\|y}\ket{\phi_k(1\|y)}}{\sqrt{2}}.
    \end{align}
    \item $\cC$ sends $\ct\coloneqq(x,y,\ket{\ct_{x,y}^b})$ to the adversary $\cA$.
    \item 
    \Erase{
    $\cA$ can classically query to the oracle $\mathcal{O}_k$, where $\mathcal{O}_k$ works as follows:}
    \begin{enumerate}
        \item \Erase{
        On input $c\in\bit$, it chooses $x',y'\gets\bit^\ell$ and generates $\ket{\ct_{x',y'}^c}$.}
        \item \Erase{
        It outputs $(x',y',\ket{\ct_{x',y'}^c})$.}
    \end{enumerate}
    {\color{red}$\cA$ receives $\ket{\tau}\coloneqq\bigotimes_{i\in[t],c\in\bit}\ket{x^i_c}\ket{y^i_c}\ket{\ct^{c}_{x^{i}_c,y^{i}_c}}$ as an auxiliary input, 
    where $t$ is the maximum number of $\cA$'s queries to $\mathcal{O}_k$ in the step 4 of Hybrid 0, and $x^i_c,y^i_c\gets\bit^\ell$ for each $i\in[t]$ and $c\in\bit$.
    When $\cA$ queries $c_i\in\bit$ to $\mathcal{O}_k$ in its $i$th query, it does not query to $\mathcal{O}_k$. Instead, it uses $\ket{x^i_c}\ket{y^i_c}\ket{\ct^c_{x^i_c,y^i_c}}$ as the output of $\mathcal{O}_k$.}
    $\cA$ sends $b'\in\bit$ to $\cC$.
    \item If $b=b'$, $\cC$ outputs $\top$. Otherwise, $\cC$ outputs $\bot$.
\end{enumerate}

Let $\mathbf{w}\coloneqq\{x^{i}_c,y^{i}_c\}_{i\in[t],c\in\bit}$, where $x^{i}_c\in\bit^\ell$ and $y^{i}_c\in\bit^\ell$ for each $i\in[t]$ and $c\in\bit$.
Let $\Pr[\top\gets\mathrm{Hybrid\:1}|k,x,y,\mathbf{w}]$ be the conditional
probability that $\cC$ outputs $\top$ given $k\gets\mathsf{UPSG.}\KeyGen(1^\secp)$ and
$x,y,\mathbf{w}$ are chosen in Hybrid 1.
We define a ``good'' set of $(k,x,y,\mathbf{w})$ as follows:
\begin{align}
    G\coloneqq\left\{(k,x,y,\mathbf{w}):\Pr[\top\gets\mathrm{Hybrid\:1}|k,x,y,\mathbf{w}]\ge\frac{1}{2}+\frac{1}{2p(\secp)}\land x\notin \mathbf{w}\land y\notin \mathbf{w}\land x\neq y\right\}.
    \label{G}
\end{align}
Let $\Pr[k,x,y,\mathbf{w}]$ be the probability that $k,x,y$ and $\mathbf{w}$ are chosen in Hybrid 1. Then, we can show the following lemma by the standard average argument. We give its proof later.
\begin{lemma}\label{lem:avg argument}
$
\sum_{(k,x,y,\mathbf{w})\in G}\Pr[k,x,y,\mathbf{w}]\ge\frac{1}{4p(\lambda)}
$
for infinitely many $\secp\in\mathbb{N}$.
\end{lemma}

Let us fix $(k,x,y,\mathbf{w})$. Moreover, assume that $(k,x,y,\mathbf{w})\in G$. Then from \cref{G}, $\cA$ of Hybrid 1 can distinguish $\ket{\ct^{0}_{x,y}}$ and $\ket{\ct^{1}_{x,y}}$ with an advantage greater than $\frac{1}{2p}$ using
the auxiliary input $\ket{\tau}$.
By using \cref{lem:HMY22}, we can construct 
a polynomial-time implementable unitary $V$\footnote{Note that this $V$ is independent of $(k,x,y,\mathbf{w})$ since, in the proof of \cref{lem:HMY22}, we use $\cA$ only as a black-box. For details, see \cite{EC:HhaMorYam23}.}
such that
\begin{align}
    \frac{1}{2p(\secp)}
    \le&
    \frac{|\bra{0\|x}\bra{\phi_k(0\|x)}\bra{\tau}V\ket{1\|y}\ket{\phi_k(1\|y)}\ket{\tau}+\bra{1\|y}\bra{\phi_k(1\|y)}\bra{\tau}V\ket{0\|x}\ket{\phi_k(0\|x)}\ket{\tau}|}{2}\\
    \le&\max\huge\{\left|\bra{0\|x}\bra{\phi_k(0\|x)}\bra{\tau}V\ket{1\|y}\ket{\phi_k(1\|y)}\ket{\tau}\right|,\left|\bra{1\|y}\bra{\phi_k(1\|y)}\bra{\tau}V\ket{0\|x}\ket{\phi_k(0\|x)}\ket{\tau}\right|\huge\}\\
    \le&\max\huge\{\|
    \bra{\phi_k(0\|x)}_\regY
    \left(V\ket{1\|y}_\regX\ket{\phi_k(1\|y)}_\regY\ket{\tau}_\regZ\right)\|,\\
    &\left\|\bra{\phi_k(1\|y)}_\regY\left(V\ket{0\|x}_\regX\ket{\phi_k(0\|x)}_\regY\ket{\tau}_\regZ\right)\right\|\huge\}.
    \label{conversion}
\end{align}
From this $V$, we construct the QPT adversary $\cB$ that breaks the security of the UPSG as follows.

\paragraph{Adversary $\cB$}

\begin{enumerate}
    \item Choose $b\gets\bit$.
    \item Choose $x,y\gets\bit^\ell$. If $x=y$, output $\bot$ and abort. Choose $x^{i}_c\gets\bit^\ell$ and $y^{i}_c\gets\bit^\ell$ for each $i\in[t]$ and $c\in\bit$. Set $\mathbf{w}\coloneqq\{x^{i}_c,y^{i}_c\}_{i\in[t],c\in\bit}$. If $x\in \mathbf{w}$ or $y\in \mathbf{w}$, output $\bot$ and abort.\label{cB:select_pahse}
    \item If $b=0$, get $\ket{0\|x}\ket{\phi_k(0\|x)}$ by querying $0\|x$ to $\mathsf{UPSG.}\Eval(k,\cdot)$. 
    If $b=1$, get $\ket{1\|y}\ket{\phi_k(1\|y)}$ by querying $1\|y$ to $\mathsf{UPSG.}\Eval(k,\cdot)$.
    \item For each $i\in[t]$ and $c\in\bit$, generate $\ket{\ct^{c}_{x^{i}_c,y^{i}_c}}$ by making the coherent query $\ket{0\|x^{i}_c}+(-1)^c\ket{1\|y^{i}_c}$ to $\mathsf{UPSG.}\Eval(k,\cdot)$. Set 
    $
        \ket{\tau}\coloneqq\bigotimes_{i\in[t],c\in\bit}\ket{x^i_c}\ket{y^i_c}\ket{\ct^{c}_{x^{i}_c,y^{i}_c}}.
        $
    \item If $b=0$, apply the unitary $V$ on $\ket{0\|x}\ket{\phi_k(0\|x)}\ket{\tau}$ and output the second register and $1\|y$. If $b=1$, apply the unitary $U$ on $\ket{1\|y}\ket{\phi_k(1\|y)}\ket{\tau}$ and output the second register and $0\|x$.
\end{enumerate}

Since $\cB$ does not abort if $(k,x,y,\mathbf{w})\in G$, the probability that the adversary $\cB$ wins is 
\begin{align}
    \Pr[\cB\:\mathrm{wins}]
    \ge&\sum_{(k,x,y,\mathbf{w})\in G}\frac{\Pr[k,x,y,\mathbf{w}]}{2}
    \left(
    \left\|
    \bra{\phi_k(0\|x)}_\regY
    \left(V\ket{1\|y}_\regX\ket{\phi_k(1\|y)}_\regY\ket{\tau}_\regZ\right)\right\|^2\right.\\
    &\left.+\left\|
    \bra{\phi_k(1\|y)}_\regY
    \left(V\ket{0\|x}_\regX\ket{\phi_k(0\|x)}_\regY\ket{\tau}_\regZ\right)\right\|^2\right)\\
    \ge&\sum_{(k,x,y,\mathbf{w})\in G}\frac{\Pr[k,x,y,\mathbf{w}]}{2}\max
    \left\{\left\|\bra{\phi_k(0\|x)}_\regY\left(V\ket{1\|y}_\regX\ket{\phi_k(1\|y)}_\regY\ket{\tau}_\regZ\right)\right\|^2\right.\\,
    &\left.\left\|\bra{\phi_k(1\|y)}_\regY\left(V\ket{0\|x}_\regX\ket{\phi_k(0\|x)}_\regY\ket{\tau}_\regZ\right)\right\|^2\right\} \label{caluculation1_in_SKE_from_UPSG}\\
    \ge&\sum_{(k,x,y,\mathbf{w})\in G}\frac{\Pr[k,x,y,\mathbf{w}]}{2}\frac{1}{4p(\secp)^2}
    \ge\frac{1}{32p(\secp)^3}\label{caluclation2_in_SKE_from_UPSG}
\end{align}
for infinitely many $\secp$,
where we have used \cref{conversion} in \cref{caluculation1_in_SKE_from_UPSG}, 
and \cref{lem:avg argument} in \cref{caluclation2_in_SKE_from_UPSG}.
This shows that $\cB$ breaks the security of the UPSG. Hence we have shown the theorem.
\end{proof}

Finally, we give a proof of \cref{lem:avg argument}

\begin{proof}[Proof of \cref{lem:avg argument}]
    We define
    \begin{align}
        T\coloneqq\left\{(k,x,y,\mathbf{w}):\Pr[\top\gets\mathrm{Hybrid\:1}|k,x,y,\mathbf{w}]\ge\frac{1}{2}+\frac{1}{2p(\secp)}\right\}.
    \end{align}
    From $\Pr[\top\gets\mathrm{Hybrid\:1}]=\Pr[\top\gets\mathrm{Hybrid\:0}]\ge\frac{1}{2}+\frac{1}{p(\secp)}$ for infinitely many $\secp$, and the definition of $T$,
    \begin{align}
    \frac{1}{2}+\frac{1}{p(\secp)}
    \le&\Pr[\top\gets\mathrm{Hybrid\:1}]\\
    =&\sum_{(k,x,y,\mathbf{w})\in T}\Pr[k,x,y,\mathbf{w}]\Pr[\top\gets\mathrm{Hybrid\:1}|k,x,y,\mathbf{w}]\\
    &+\sum_{(k,x,y,\mathbf{w})\notin T}\Pr[k,x,y,\mathbf{w}]\Pr[\top\gets\mathrm{Hybrid\:1}|k,x,y,\mathbf{w}]\\
    <&\sum_{(k,x,y,\mathbf{w})\in T}\Pr[k,x,y,\mathbf{w}]+\left(\frac{1}{2}+\frac{1}{2p(\secp)}\right)\sum_{(k,x,y,\mathbf{w})\notin T}\Pr[k,x,y,\mathbf{w}]\\
    \le&\sum_{(k,x,y,\mathbf{w})\in T}\Pr[k,x,y,\mathbf{w}]+\frac{1}{2}+\frac{1}{2p(\secp)}
    \end{align}
    for infinitely many $\secp$,
    which means
    \begin{align}
        \sum_{(k,x,y,\mathbf{w})\in T}\Pr[k,x,y,\mathbf{w}]\ge\frac{1}{2p(\secp)}\label{caluclation_for_T}
    \end{align}
    for infinitely many $\secp$.
    Thus,
    \begin{align}
    \sum_{(k,x,y,\mathbf{w})\in G}\Pr[k,x,y,\mathbf{w}]=&1-\sum_{(k,x,y,\mathbf{w})\notin G}\Pr[k,x,y,\mathbf{w}]\\
    =&1-\sum_{(k,x,y,\mathbf{w})\notin T\lor x\in\mathbf{w}\lor y\in\mathbf{w}\lor x= y}\Pr[k,x,y,\mathbf{w}]\\
    \ge&1-\left(\sum_{(k,x,y,\mathbf{w})\notin T}\Pr[k,x,y,\mathbf{w}]+\sum_{x\in\mathbf{w}\lor y\in\mathbf{w}\lor x= y}\Pr[k,x,y,\mathbf{w}]\right)\\
    \ge&\frac{1}{2p(\secp)}-\negl(\secp)
    \ge\frac{1}{4p(\secp)}
    \end{align}
    for infinitely many $\secp$,
    where the first inequality follows from the union bound and in the second inequality we have used \cref{caluclation_for_T} and $\sum_{x\in\mathbf{w}\lor y\in\mathbf{w}\lor x= y}\Pr[k,x,y,\mathbf{w}]=\negl(\lambda)$ since $x$, $y$ and each element of $\mathbf{w}$ is selected independently and uniformly at random.
\end{proof}

\section{MACs with Unclonable Tags}
\label{sec:unclonable_tag}
In this section, we define MACs with unclonable tags and construct it from UPSGs. 

\subsection{Definition}
First, we give the definition of the standard EUF-CMA-secure MACs. However, in this paper, we consider more general case where the tags could be quantum states.
MACs with classical tags can be considered as a special case where the tags are computational-basis states.
\begin{definition}[EUF-CMA-Secure MACs]
An EUF-CMA-secure MAC is a set $(\KeyGen,\Tag,\Ver)$ of QPT algorithms such that
\begin{itemize}
    \item 
    $\KeyGen(1^\secp)\to\sigk:$
    It is a QPT algorithm that, on input the security parameter $\secp$, outputs a classical key $\sigk$.
    \item 
    $\Tag(\sigk,m)\to\tau:$
    It is a QPT algorithm that, on input $\sigk$ and a classical message $m$, outputs an $n$-qubit quantum state $\tau$.
    \item 
    $\Ver(\sigk,m,\rho)\to\top/\bot:$
    It is a QPT algorithm that, on input $\sigk$, $m$, and a quantum state $\rho$, outputs $\top/\bot$.
\end{itemize}
   We require the following two properties.
   
   \paragraph{Correctness:}
   For any $m$,
   \begin{align}
    \Pr\left[\top\gets\Ver(\sigk,m,\tau):
    \begin{array}{rr}
    \sigk\gets\KeyGen(1^\secp)\\
    \tau\gets\Tag(\sigk,m)
    \end{array}
    \right]\ge1-\negl(\secp). 
   \end{align}

   \paragraph{EUF-CMA security:}
   For any QPT adversary $\cA$,
   \begin{align}
    \Pr\left[\top\gets\Ver(\sigk,m^*,\rho):
    \begin{array}{rr}
    \sigk\gets\KeyGen(1^\secp)\\
    (m^*,\rho)\gets\cA^{\mathsf{Tag}(\sigk,\cdot)}(1^{\lambda})
    \end{array}
    \right]\le\negl(\secp), 
   \end{align}
   where $\cA$ queries the oracle only classically, and $\cA$ is not allowed to query $m^*$.

\end{definition}

The following corollary is straightforward from the definition of UPSGs.
\begin{corollary}
    \label{coro:MAC_from_UPSG}
If UPSGs exist, then EUF-CMA-secure MACs exist.    
\end{corollary}

\begin{proof}[Proof of \cref{coro:MAC_from_UPSG}]
    Let $(\KeyGen',\Eval')$ be a UPSG. We construct EUF-CMA-secure MAC $(\KeyGen,\Tag,\Ver)$ as follows:
    \begin{itemize}
        \item $\KeyGen(1^\secp)\to\sigk:$ Run $k\gets\KeyGen'(1^\secp)$ and output it as $\sigk$.
        \item $\Tag(\sigk,m)\to\tau:$ Parse $\sigk=k$. Run $\ket{\phi_k(m)}\gets\Eval'(k,m)$ and output it as $\tau$.
        \item $\Ver(\sigk,m,\rho)\to\top/\bot:$ Parse $\sigk=k$. Project $\rho$ onto $\ket{\phi_k(m)}\bra{\phi_k(m)}$. If the projection is successful, output $\top$. Otherwise, output $\bot$.
    \end{itemize}
    The correctness is clear. The EUF-CMA-security follows from the unpredictability of UPSG.
\end{proof}

Next, we define MACs with unclonable tags.

\begin{definition}[MACs with Unclonable Tags]
   Let $(\KeyGen,\Tag,\Ver)$ be an EUF-CMA-secure MAC. If $(\KeyGen,\Tag,\Ver)$ satisfies the following property (which we call unclonability), we call it MAC with unclonable tags: 
   For any QPT adversary $\cA$ and any polynomials $t$ and $\ell$,
   \begin{align}
    \Pr\left[
    \mathsf{Count}(\sigk,m^*,\xi)\ge t+1
    :
    \begin{array}{rr}
    \sigk\gets\KeyGen(1^\secp)\\
    (m^*,\st)\gets\cA^{\mathsf{Tag}(\sigk,\cdot)}(1^{\lambda})\\
    \tau^{\otimes t}\gets\Tag(\sigk,m^*)^{\otimes t}\\
    \xi\gets\cA^{\Tag(\sigk,\cdot)}(\tau^{\otimes t},\st)
    \end{array}
    \right]\le\negl(\secp), 
    \label{unclonability}
   \end{align}
   where $\cA$ queries the oracle only classically, and $\cA$ is not allowed to query $m^*$. 
    $\tau^{\otimes t}\gets\Tag(\sigk,m^*)^{\otimes t}$ means that $\Tag$ algorithm is run $t$ times and $t$ copies of $\tau$ are generated.
   $\xi$ is a quantum state on $\ell$ registers, $\regR_1,...,\regR_\ell$, each of which is of $n$ qubits. Here, $\mathsf{Count}(\sigk,m^*,\xi)$ is the following QPT algorithm: for each $j\in[\ell]$, it takes the state on $\regR_j$ as input, and runs $\Ver(\sigk,m^*,\cdot)$ to get $\top$ or $\bot$. Then, it outputs the total number of $\top$.
\end{definition}

\begin{remark}
EUF-CMA security is automatically implied by the unclonability, \cref{unclonability}.\footnote{The proof is easy. Let $\cA$ be a QPT adversary that breaks the EUF-CMA security,
which outputs $(m^*,\rho)$. Then the QPT adversary $\cB$ that breaks the unclonability is constructed as follows:
it first simulates $\cA$ to get $(m^*,\rho)$.
It then sends $m^*$ to the challenger to get its tag $\tau$.
It finally sends $\tau$ and $\rho$ to the challenger, both of which are accepted as valid tags.
} 
\end{remark}

\subsection{Construction from UPSGs}
In this subsection, we construct MACs with unclonable tags from EUF-CMA-secure MACs and IND-CPA-secure SKE schemes. 

\begin{theorem}\label{thm:unclonable_MAC_from_UPSG_and_SKE}
If EUF-CMA-secure MACs (secure against classically querying QPT adversaries) 
and IND-CPA-secure SKE schemes for classical messages (secure against classically querying QPT adversaries) exist, then MACs with unclonable tags exist.
\end{theorem}
Because
EUF-CMA-secure MACs (secure against classically querying QPT adversaries) 
can be constructed from UPSG (\cref{coro:MAC_from_UPSG}), and
IND-CPA-secure SKE schemes for classical messages (secure against classically querying QPT adversaries) can be constructed from UPSGs (\cref{thm:SKE_from_UPSG}), we have the following corollary:
\begin{corollary}\label{coro:unclonable_MAC_from_UPSG}
If UPSGs exist, then MACs with unclonable tags exist.    
\end{corollary}

\begin{proof}[Proof of \cref{thm:unclonable_MAC_from_UPSG_and_SKE}]
    Let $(\mathsf{MAC.}\KeyGen,\mathsf{MAC.}\Tag,\mathsf{MAC.}\Ver)$ be an EUF-CMA-secure MAC secure against classically querying QPT adversaries
    and 
    $(\mathsf{SKE.}\KeyGen,\allowbreak\mathsf{SKE.}\Enc, \mathsf{SKE.}\Dec)$ be an IND-CPA-secure SKE scheme for quantum messages.
   (From \cref{lem:SKE_for_classical_imply_SKE_for_quantum}, such SKE schemes exist if SKE schemes for classical messages secure against classically
   querying QPT adversaries exist.) 
    We construct a MAC with unclonable tags $(\KeyGen,\Tag,\Ver)$ as follows:
    \begin{itemize}
        \item $\KeyGen(1^\secp)\to \sigk':$ Run $\sk\gets\mathsf{SKE.}\KeyGen(1^\secp)$ and $\sigk\gets\mathsf{MAC.}\KeyGen(1^\secp)$.
        Output $\sigk'\coloneqq(\sk,\sigk)$.
        \item $\Tag(\sigk',m)\to\tau':$ Parse $\sigk'=(\sk,\sigk)$. It does the following:
        \begin{enumerate}
            \item Choose $x,\theta\gets\bit^\secp$ and generate $\ket{x}_{\theta}$. Here, $\ket{x}_{\theta}\coloneqq\bigotimes_{i\in[\secp]}H^{\theta^i}\ket{x^i}$, 
            where $H$ is the Hadamard gate, and $x^i$ and $\theta^i$ denote the $i$'th bit of $x$ and $\theta$, respectively.
            \item Run $\tau\gets\mathsf{MAC.}\Tag(\sigk,m\|x\|\theta)$.
            \item Run $\ct\gets\mathsf{SKE.}\Enc(\sk,\ket{x\|\theta}\bra{x\|\theta}\otimes\tau)$.
        \end{enumerate}
        Output $\tau'\coloneqq |x\rangle\langle x|_{\theta}\otimes\ct$.
        \item $\Ver(\sigk',m,\rho)\to\top/\bot:$ Parse $\sigk'=(\sk,\sigk)$. 
        Let $\rho$ be a state on two registers $\regA$ and $\regC$. (If $\rho$ is honestly generated, $\rho_{\regA,\regC}=(|x\rangle\langle x|_\theta)_\regA\otimes\ct_\regC$.) 
        It does the following:
        \begin{enumerate}
            \item Run $\mathsf{SKE.}\Dec(\sk,\cdot)$ on the register $\regC$ to get another state $\rho'_{\regA,\regM}$ on the registers $\regA$ and $\regM$.
            \item Measure the first $2\secp$ qubits of $\regM$ in the computational basis to get the result $x'\|\theta'$.  
            \item Run $\mathsf{MAC.}\Ver(\sigk,m\|x'\|\theta',\cdot)$ on the remaining qubits of the register $\regM$ to get $v\in\{\top,\bot\}$. 
            Project the register $\regA$ onto $\ket{x'}_{\theta'}$. If the projection is successful and $v=\top$, output $\top$. 
            Otherwise, output $\bot$.
        \end{enumerate}
       
    \end{itemize}
    The correctness is clear. Since the unclonablity implies EUF-CMA security, it suffices to show our construction satisfies the unclonability. 
    Let $t$ and $\ell$ be polynomials. We define the Hybrid 0 as follows, which is the original security game of unclonability between the challenger $\cC$ and QPT adversary $\cA$.
    
    \paragraph{Hybrid 0}
    \begin{enumerate}
    \item The challenger $\cC$ runs $\sk\gets\mathsf{SKE.}\KeyGen(1^{\lambda})$ and $\sigk\gets\mathsf{MAC.}\KeyGen(1^\secp)$.
    \item The adversary $\cA$ sends $m^*$ to $\cC$, where $\cA$ can make classical queries to the oracle $\mathcal{O}_{\sk,\sigk}$ and does not query $m^*$. Here, $\mathcal{O}_{\sk,\sigk}$ takes a bit string $m$ as input and works as follows:
    \begin{enumerate}
        \item Choose $x,\theta\gets\bit^\secp$ and generate $\ket{x}_{\theta}$.
        \item Run $\tau\gets\mathsf{MAC.}\Tag(\sigk,m\|x\|\theta)$ and $\ct\gets\mathsf{SKE.}\Enc(\sk,\ket{x\|\theta}\bra{x\|\theta}\otimes\tau)$.
        \item Output $\ket{x}\bra{x}_{\theta}\otimes\ct$.
    \end{enumerate}
    \item For each $i\in[t]$, $\cC$ does the following.
    \begin{enumerate}
        \item Choose $x_i,\theta_i\gets\bit^\secp$ and generate $\ket{x_i}_{\theta_i}$.
        \item Run $\tau_i\gets\mathsf{MAC.}\Tag(\sigk,m^*\|x_i\|\theta_i)$ and $\ct_i\gets\mathsf{SKE.}\Enc(\sk,\ket{x_i\|\theta_i}\bra{x_i\|\theta_i}\otimes\tau_i)$.
    \end{enumerate} 
    \item $\cC$ sends $\{\ket{x_i}\bra{x_i}_{\theta_i}\otimes\ct_i\}_{i\in[t]}$ to $\cA$.
    \item $\cA$ sends $\xi_{\regR_1,...,\regR_\ell}$, where $\cA$ can make classical queries to the oracle $\mathcal{O}_{\sk,\sigk}$ and does not query $m^*$. 
    Here $\regR_j$ has two registers $\regA_j$ and $\regC_j$ for each $j\in[\ell]$.
    \item For each $j\in[\ell]$, $\cC$ does the following.\label{step:ver_hyb0}
    \begin{enumerate}
            \item Run $\mathsf{SKE.}\Dec(\sk,\cdot)$ on the register $\regC_j$ to get the state on the registers $\regA_j$ and $\regM_j$.
            \item Measure the first $2\secp$ qubits of the register $\regM_j$ in the computational basis to get the result $x'_j\|\theta'_j$.  
            \item Run $\mathsf{MAC.}\Ver(\sigk,m^*\|x'_j\|\theta'_j,\cdot)$ on the remaining qubits of the register $\regM_j$ to get $v_j\in\{\top,\bot\}$. 
            \item Project the register $\regA_j$ onto $\ket{x'_j}_{\theta'_j}$. 
            \item If the projection is successful and $v_j=\top$, set $w_j \coloneqq 1$. 
            Otherwise, set $w_j \coloneqq 0$.
    \end{enumerate}
    \item 
    If $\sum_{j=1}^\ell w_j\ge t+1$, $\cC$ outputs $\top$. Otherwise, $\cC$ outputs $\bot$. 
    \end{enumerate}

    To show the theorem, let us assume that there exists a QPT adversary $\cA$ such that
    $\Pr[\top\gets\mbox{Hybrid 0}]\ge\frac{1}{\poly(\secp)}$ for infinitely many $\secp$.
    Our goal is to construct an adversary that breaks the security of the Wiesner money scheme from $\cA$. 
    For that goal, we want to make sure that two copies of $|x\rangle_{\theta}$ are generated when $\cC$ outputs $\top$. 
    The next Hybrid 1 ensures such a situation, and the hop from Hybrid 0 to 1 can be done
    by invoking the EUF-CMA security of the MAC.\footnote{This is actually a well-known technique to construct a full money from a mini-scheme~\cite{STOC:AarChr12}.}
    \paragraph{Hybrid 1}

    \begin{enumerate}
    \item The challenger $\cC$ runs $\sk\gets\mathsf{SKE.}\KeyGen(1^{\lambda})$ and $\sigk\gets\mathsf{MAC.}\KeyGen(1^\secp)$.
    \item The adversary $\cA$ sends $m^*$ to $\cC$, where $\cA$ can make classical queries to the oracle $\mathcal{O}_{\sk,\sigk}$ and does not query $m^*$. Here, $\mathcal{O}_{\sk,\sigk}$ takes a bit string $m$ as input and works as follows:
    \begin{enumerate}
        \item Choose $x,\theta\gets\bit^\secp$ and generate $\ket{x}_{\theta}$.
        \item Run $\tau\gets\mathsf{MAC.}\Tag(\sigk,m\|x\|\theta)$ and $\ct\gets\mathsf{SKE.}\Enc(\sk,\ket{x\|\theta}\bra{x\|\theta}\otimes\tau)$.
        \item Output $\ket{x}\bra{x}_{\theta}\otimes\ct$.
    \end{enumerate}\label{step:challenge_query_hyb1}
    \item For each $i\in[t]$, $\cC$ does the following.
    \begin{enumerate}
        \item Choose $x_i,\theta_i\gets\bit^\secp$ and generate $\ket{x_i}_{\theta_i}$.
        \item Run $\tau_i\gets\mathsf{MAC.}\Tag(\sigk,m^*\|x_i\|\theta_i)$ and $\ct_i\gets\mathsf{SKE.}\Enc(\sk,\ket{x_i\|\theta_i}\bra{x_i\|\theta_i}\otimes\tau_i)$.
    \end{enumerate} 
    \item $\cC$ sends $\{\ket{x_i}\bra{x_i}_{\theta_i}\otimes\ct_i\}_{i\in[t]}$ to $\cA$.
    \item $\cA$ sends $\xi_{\regR_1,...,\regR_\ell}$, where $\cA$ can make classical queries to the oracle $\mathcal{O}_{\sk,\sigk}$ and does not query $m^*$. 
    Here $\regR_j$ has two registers $\regA_j$ and $\regC_j$ for each $j\in[\ell]$.\label{step:challenge_tag_hyb1}
    \item For each $j\in[\ell]$, $\cC$ does the followings.
    \begin{enumerate}
            \item Run $\mathsf{SKE.}\Dec(\sk,\cdot)$ on the register $\regC_j$ to get the state on the registers $\regA_j$ and $\regM_j$.
            \item Measure the first $2\secp$ qubits of the register $\regM_j$ in the computational basis to get the result $x'_j\|\theta'_j$.  
            \item Run $\mathsf{MAC.}\Ver(\sigk,m^*\|x'_j\|\theta'_j,\cdot)$ on the remaining qubits of the register $\regM_j$ to get $v_j\in\{\top,\bot\}$. 
            \item Project the register $\regA_j$ onto $\ket{x'_j}_{\theta'_j}$. 
            \item If the projection is successful and $v_j=\top$, set $w_j \coloneqq 1$. 
            Otherwise, set $w_j \coloneqq 0$.
    \end{enumerate}
    \item 
    If $\sum^\ell_{j=1}w_j\ge t+1$ {\color{red}and the event $E$ does not occur}, then 
    $\cC$ outputs $\top$. Otherwise, $\cC$ outputs $\bot$. 
    {\color{red} Here $E$ is the event defined as follows:} 
    \begin{itemize}
        \item {\color{red}\mbox{Event $E$}: \mbox{there exists $j\in[\ell]$ such that $(x'_j,\theta'_j)\not\in\{(x_i,\theta_i)\}_{i\in[t]}$ and $w_j=1$.}}
    \end{itemize}
    \end{enumerate}

    \begin{lemma}\label{lem:hyb1}
        $\Pr[\top\gets\rm{Hybrid\:0}]\le
        \Pr[\top\gets\rm{Hybrid\:1}]+\negl(\secp)$.
    \end{lemma}

    \begin{proof}[Proof of \cref{lem:hyb1}]
        We can show 
        \begin{align}
        \Pr[E]\le\negl(\secp) 
        \label{E_small}
        \end{align}
        whose proof is given later. If $\Pr[E]\le\negl(\secp)$,
        \begin{align}
            \Pr[\top\gets\rm{~Hybrid\:0}]
            &=\Pr[\top\gets\mathrm{~Hybrid\:0}\wedge E]+\Pr[\top\gets\mathrm{~Hybrid\:0}\wedge \bar{E}]\\
            &\le\negl(\secp)+\Pr[\top\gets\rm{~Hybrid\:1}],
        \end{align}
        which shows the lemma.

        Let us show \cref{E_small}. Assume that $\Pr[E]\ge\frac{1}{\poly(\secp)}$ for infinitely many $\secp\in\mathbb{N}$. 
        Then the following QPT adversary $\cB$ breaks the EUF-CMA security of the MAC:
        \begin{enumerate}
            \item The adversary $\cB$ runs $\sk\gets\mathsf{SKE.}\KeyGen(1^\secp)$.
            \item $\cB$ simulates the interaction between $\cC$ and $\cA$ in Hybrid 1 by querying to $\mathsf{MAC.}\Tag(\sigk,\cdot)$ up to the step \ref{step:challenge_tag_hyb1}. 
            Then, $\cB$ gets a classical message $m^*$ and a state $\xi$ on the registers $\regR_1,...\regR_\ell$, where $m^*$ is a challenge message that $\cA$ sends to $\cC$ in the step \ref{step:challenge_query_hyb1}. 
            Here $\regR_j$ has two registers $\regA_j$ and $\regC_j$ for each $j\in[\ell]$.
            \item For each $j\in[\ell]$, $\cB$ does the following:
            \begin{enumerate}
            \item Run $\mathsf{SKE.}\Dec(\sk,\cdot)$ on the register $\regC_j$ to get the state on the registers $\regA_j$ and $\regM_j$.
            \item Measure the first $2\secp$ qubits of the register $\regM_j$ in the computational basis to get the result $x'_j\|\theta'_j$. 
            \end{enumerate}
            \item $\cB$ chooses $j^*\gets [\ell]$. If $(x'_{j^*},\theta'_{j^*})\in\{(x_i,\theta_i)\}_{i\in[t]}$, $\cB$ aborts. Otherwise, $\cB$ outputs $m^*\|x'_{j^*}\|\theta'_{j^*}$ and the all qubits of the register $\regM_{j^*}$ except for the first $2\secp$-qubits.
        \end{enumerate}
        It is clear that $\cB$ does not query $m^*\|x'_{j^*}\|\theta'_{j^*}$.
        Let $\Pr[\cB\:\rm{wins}]$ be the probability that $\cB$ wins the above security game of EUF-CMA security. Then, we have
        $
            \Pr[\cB\:\mathrm{wins}]\ge\frac{1}{\ell}\Pr[E].
            $
        Therefore, $\cB$ breaks the EUF-CMA security if $\Pr[E]\ge\frac{1}{\poly(\secp)}$ for infinitely many $\secp\in\mathbb{N}$. This means $\Pr[E]\le\negl(\secp)$.
    \end{proof}

    If $\Pr[\top\gets\mbox{Hybrid 1}]\ge\frac{1}{\poly(\secp)}$ for infinitely many $\secp$,
    at least two copies of $|x\rangle_\theta$ for some $x$ and $\theta$ should be generated due to the pigeonhole principle.
    In the following Hybrid 2, we randomly guess the indexes of such states. 
    Then we have the following lemma.
        \begin{lemma}\label{lem:hyb2}
        
     $\Pr[\top\gets\mathrm{Hybrid\:2}]\ge\frac{1}{t}\Pr[\top\gets\mathrm{Hybrid\:1}]$.
    \end{lemma}

    \paragraph{Hybrid 2}
    \begin{enumerate}
    \item The challenger $\cC$ runs $\sk\gets\mathsf{SKE.}\KeyGen(1^{\lambda})$ and $\sigk\gets\mathsf{MAC.}\KeyGen(1^\secp)$.
    \item The adversary $\cA$ sends $m^*$ to $\cC$, where $\cA$ can make classical queries to the oracle $\mathcal{O}_{\sk,\sigk}$ and does not query $m^*$. Here, $\mathcal{O}_{\sk,\sigk}$ takes a bit string $m$ as input and works as follows:
    \begin{enumerate}
        \item Choose $x,\theta\gets\bit^\secp$ and generate $\ket{x}_{\theta}$.
        \item Run $\tau\gets\mathsf{MAC.}\Tag(\sigk,m\|x\|\theta)$ and $\ct\gets\mathsf{SKE.}\Enc(\sk,\ket{x\|\theta}\bra{x\|\theta}\otimes\tau)$.
        \item Output $\ket{x}\bra{x}_{\theta}\otimes\ct$.
    \end{enumerate}
    \item {\color{red} $\cC$ chooses $i^*\gets[t]$.} For each $i\in[t]$, $\cC$ does the following.
    \begin{enumerate}
        \item Choose $x_i,\theta_i\gets\bit^\secp$ and generate $\ket{x_i}_{\theta_i}$.
        \item Run $\tau_i\gets\mathsf{MAC.}\Tag(\sigk,m^*\|x_i\|\theta_i)$ and $\ct_i\gets\mathsf{SKE.}\Enc(\sk,\ket{x_i\|\theta_i}\bra{x_i\|\theta_i}\otimes\tau_i)$.
    \end{enumerate} 
    \item $\cC$ sends $\{\ket{x_i}\bra{x_i}_{\theta_i}\otimes\ct_i\}_{i\in[t]}$ to $\cA$.
    \item $\cA$ sends $\xi_{\regR_1,...,\regR_\ell}$, where $\cA$ can make classical queries to the oracle $\mathcal{O}_{\sk,\sigk}$ and does not query $m^*$. 
    Here $\regR_j$ has two registers $\regA_j$ and $\regC_j$ for each $j\in[\ell]$.
    \item For each $j\in[\ell]$, $\cC$ does the followings.
    \begin{enumerate}
            \item Run $\mathsf{SKE.}\Dec(\sk,\cdot)$ on the register $\regC_j$ to get the state on the registers $\regA_j$ and $\regM_j$.
            \item Measure the first $2\secp$ qubits of the register $\regM_j$ in the computational basis to get the result $x'_j\|\theta'_j$.  
            \item Run $\mathsf{MAC.}\Ver(\sigk,m^*\|x'_j\|\theta'_j,\cdot)$ on the remaining qubits of the register $\regM_j$ to get $v_j\in\{\top,\bot\}$. 
            \item Project the register $\regA_j$ onto $\ket{x'_j}_{\theta'_j}$. 
            \item If the projection is successful and $v_j=\top$ {\color{red} and $(x'_j,\theta'_j)=(x_{i^*},\theta_{i^*})$}, set $w_j \coloneqq 1$. 
            Otherwise, set $w_j \coloneqq 0$.
    \end{enumerate}
    \item 
    \Erase{If $\sum^\ell_{j=1}w_j\ge t+1$ and the event $E$ does not occur,}
    {\color{red}If $\sum^\ell_{j=1}w_j\ge2$}, $\cC$ outputs $\top$. Otherwise, $\cC$ outputs $\bot$. \Erase{ Here $E$ is the event defined as follows:} 
    \begin{itemize}
        \item \Erase{\mbox{Event $E$}: \Erase{\mbox{there exists $j\in[\ell]$ such that $(x'_j,\theta'_j)\not\in\{(x_i,\theta_i)\}_{i\in[t]}$ and $w_j=1$.}}}
    \end{itemize}
    \end{enumerate}

Let us define Hybrid 3 as follows.
The following lemma is straightforward.
    \begin{lemma}\label{lem:hyb3}
    $\Pr[\top\gets\mathrm{Hybrid\:3}]\ge\Pr[\top\gets\mathrm{Hybrid\:2}]$.
    \end{lemma}

    \paragraph{Hybrid 3}
    \begin{enumerate}
    \item The challenger $\cC$ runs $\sk\gets\mathsf{SKE.}\KeyGen(1^{\lambda})$ and $\sigk\gets\mathsf{MAC.}\KeyGen(1^\secp)$.
    \item The adversary $\cA$ sends $m^*$ to $\cC$, where $\cA$ can make classical queries to the oracle $\mathcal{O}_{\sk,\sigk}$ and does not query $m^*$. Here, $\mathcal{O}_{\sk,\sigk}$ takes a bit string $m$ as input and works as follows:
    \begin{enumerate}
        \item Choose $x,\theta\gets\bit^\secp$ and generate $\ket{x}_{\theta}$.
        \item Run $\tau\gets\mathsf{MAC.}\Tag(\sigk,m\|x\|\theta)$ and $\ct\gets\mathsf{SKE.}\Enc(\sk,\ket{x\|\theta}\bra{x\|\theta}\otimes\tau)$.
        \item Output $\ket{x}\bra{x}_{\theta}\otimes\ct$.
    \end{enumerate}
    \item $\cC$ chooses $i^*\gets[t]$. For each $i\in[t]$, $\cC$ does the following.
    \begin{enumerate}
        \item Choose $x_i,\theta_i\gets\bit^\secp$ and generate $\ket{x_i}_{\theta_i}$.
        \item Run $\tau_i\gets\mathsf{MAC.}\Tag(\sigk,m^*\|x_i\|\theta_i)$ and $\ct_i\gets\mathsf{SKE.}\Enc(\sk,\ket{x_i\|\theta_i}\bra{x_i\|\theta_i}\otimes\tau_i)$.
    \end{enumerate} 
    \item $\cC$ sends $\{\ket{x_i}\bra{x_i}_{\theta_i}\otimes\ct_i\}_{i\in[t]}$ to $\cA$.
    \item $\cA$ sends $\xi_{\regR_1,...,\regR_\ell}$, where $\cA$ can make classical queries to the oracle $\mathcal{O}_{\sk,\sigk}$ and does not query $m^*$. 
    Here $\regR_j$ has two registers $\regA_j$ and $\regC_j$ for each $j\in[\ell]$.
    \item For each $j\in[\ell]$, $\cC$ does the followings.\label{step:ver_hyb3}
    \begin{enumerate}
            \item Run $\mathsf{SKE.}\Dec(\sk,\cdot)$ on the register $\regC_j$ to get the state on the registers $\regA_j$ and $\regM_j$.
            \item \Erase{Measure the first $2\secp$ qubits of the register $\regM_j$ in the computational basis to get the result $x'_j\|\theta'_j$.}
            {\color{red} Set $(x'_j,\theta'_j)\coloneqq(x_{i^*},\theta_{i^*})$.}
            \item Run $\mathsf{MAC.}\Ver(\sigk,m^*\|x'_j\|\theta'_j,\cdot)$ on the remaining qubits of the register $\regM_j$ to get $v_j\in\{\top,\bot\}$. 
            \item Project the register $\regA_j$ onto $\ket{x'_j}_{\theta'_j}$. 
            \item If the projection is successful and $v_j=\top$ \Erase{and $(x'_j,\theta'_j)=(x_{i^*},\theta_{i^*})$}, set $w_j \coloneqq 1$. 
            Otherwise, set $w_j \coloneqq 0$.
    \end{enumerate}
    \item 
    If $\sum^\ell_{j=1}w_j\ge 2$, $\cC$ outputs $\top$. Otherwise, $\cC$ outputs $\bot$. 
    \end{enumerate}

    Now in Hybrid 3 two copies of $|x_{i^*}\rangle_{\theta_{i^*}}$ are generated.
    In order to use it to break the security of the Wiesner money scheme, we have to remove the classical description of BB84 states ``hidden'' in the ciphertexts.
 If we introduce Hybrid 4 as follows, the following lemma is straightforward.
     \begin{lemma}\label{lem:hyb4}
        $\Pr[\top\gets\rm{Hybrid\:4}]\ge\Pr[\top\gets\rm{Hybrid\:3}]$.
    \end{lemma}

    \paragraph{Hybrid 4}
    \begin{enumerate}
    \item The challenger $\cC$ runs $\sk\gets\mathsf{SKE.}\KeyGen(1^{\lambda})$ and $\sigk\gets\mathsf{MAC.}\KeyGen(1^\secp)$.
    \item The adversary $\cA$ sends $m^*$ to $\cC$, where $\cA$ can make classical queries to the oracle $\mathcal{O}_{\sk,\sigk}$ and does not query $m^*$. Here, $\mathcal{O}_{\sk,\sigk}$ takes a bit string $m$ as input and works as follows:
    \begin{enumerate}
        \item Choose $x,\theta\gets\bit^\secp$ and generate $\ket{x}_{\theta}$.
        \item Run $\tau\gets\mathsf{MAC.}\Tag(\sigk,m\|x\|\theta)$ and $\ct\gets\mathsf{SKE.}\Enc(\sk,\ket{x\|\theta}\bra{x\|\theta}\otimes\tau)$.
        \item Output $\ket{x}\bra{x}_{\theta}\otimes\ct$.
    \end{enumerate}
    \item $\cC$ chooses $i^*\gets[t]$. For each $i\in[t]$, $\cC$ does the following.
    \begin{enumerate}
        \item Choose $x_i,\theta_i\gets\bit^\secp$ and generate $\ket{x_i}_{\theta_i}$.
        \item Run $\tau_i\gets\mathsf{MAC.}\Tag(\sigk,m^*\|x_i\|\theta_i)$ and $\ct_i\gets\mathsf{SKE.}\Enc(\sk,\ket{x_i\|\theta_i}\bra{x_i\|\theta_i}\otimes\tau_i)$.
    \end{enumerate} 
    \item $\cC$ sends $\{\ket{x_i}\bra{x_i}_{\theta_i}\otimes\ct_i\}_{i\in[t]}$ to $\cA$.
    \item $\cA$ sends $\xi_{\regR_1,...,\regR_\ell}$, where $\cA$ can make classical queries to the oracle $\mathcal{O}_{\sk,\sigk}$ and does not query $m^*$. 
    Here $\regR_j$ has two registers $\regA_j$ and $\regC_j$ for each $j\in[\ell]$.
    \item For each $j\in[\ell]$, $\cC$ does the followings.
    \begin{enumerate}
            \item \Erase{Run $\mathsf{SKE.}\Dec(\sk,\cdot)$ on the register $\regC_j$ to get the state on the registers $\regA_j$ and $\regM_j$.}
            {\color{red} Does nothing in this step.}
            \item Set $(x'_j,\theta'_j)\coloneqq(x_{i^*},\theta_{i^*})$.
            \item \Erase{Run $\mathsf{MAC.}\Ver(\sigk,m^*\|x'_j\|\theta'_j,\cdot)$ on the remaining qubits of the register $\regM_j$ to get $v_j\in\{\top,\bot\}$.}
            {\color{red} Does nothing in this step.}
            \item Project the register $\regA_j$ onto $\ket{x'_j}_{\theta'_j}$. 
            \item If the projection is successful \Erase{and $v_j=\top$}, set $w_j \coloneqq 1$.
            Otherwise, set $w_j \coloneqq 0$.
    \end{enumerate}
    \item 
    If $\sum^\ell_{j=1}w_j\ge 2$, $\cC$ outputs $\top$. Otherwise, $\cC$ outputs $\bot$. 
    \end{enumerate}

    Now, we are ready to remove the information about the BB84 state from $\ct_i$ by invoking IND-CPA security. We formalize it as Hybrid 5. 

    \paragraph{Hybrid 5}
    \begin{enumerate}
    \item The challenger $\cC$ runs $\sk\gets\mathsf{SKE.}\KeyGen(1^{\lambda})$ and $\sigk\gets\mathsf{MAC.}\KeyGen(1^\secp)$.
    \item The adversary $\cA$ sends $m^*$ to $\cC$, where $\cA$ can make classical queries to the oracle $\mathcal{O}_{\sk,\sigk}$ and does not query $m^*$. Here, $\mathcal{O}_{\sk,\sigk}$ takes a bit string $m$ as input and works as follows:
    \begin{enumerate}
        \item Choose $x,\theta\gets\bit^\secp$ and generate $\ket{x}_{\theta}$.
        \item Run $\tau\gets\mathsf{MAC.}\Tag(\sigk,m\|x\|\theta)$ and $\ct\gets\mathsf{SKE.}\Enc(\sk,\ket{x\|\theta}\bra{x\|\theta}\otimes\tau)$.
        \item Output $\ket{x}\bra{x}_{\theta}\otimes\ct$.
    \end{enumerate}\label{step:test_query_hyb5}
    \item $\cC$ chooses $i^*\gets[t]$. For each $i\in[t]$, $\cC$ does the following.
    \begin{enumerate}
        \item Choose $x_i,\theta_i\gets\bit^\secp$ and generate $\ket{x_i}_{\theta_i}$.
        \item Run $\tau_i\gets\mathsf{MAC.}\Tag(\sigk,m^*\|x_i\|\theta_i)$ \Erase{and $\ct_i\gets\mathsf{SKE.}\Enc(\sk,\ket{x_i\|\theta_i}\bra{x_i\|\theta_i}\otimes\tau_i)$.} {\color{red}Run $\ct_i\gets\mathsf{SKE.}\Enc(\sk,\ket{0...0}\bra{0...0})$.}\label{step:enc_hyb5}
    \end{enumerate} 
    \item $\cC$ sends $\{\ket{x_i}\bra{x_i}_{\theta_i}\otimes\ct_i\}_{i\in[t]}$ to $\cA$.\label{step:send_hyb5}
    \item $\cA$ sends $\xi_{\regR_1,...,\regR_\ell}$, where $\cA$ can make classical queries to the oracle $\mathcal{O}_{\sk,\sigk}$ and does not query $m^*$. 
    Here $\regR_j$ has two registers $\regA_j$ and $\regC_j$ for each $j\in[\ell]$.
    \item For each $j\in[\ell]$, $\cC$ does the followings.
    \begin{enumerate}
            \item Does nothing in this step.
            \item Set $(x'_j,\theta'_j)\coloneqq(x_{i^*},\theta_{i^*})$.
            \item Does nothing in this step.
            \item Project the register $\regA_j$ onto $\ket{x'_j}_{\theta'_j}$. 
            \item If the projection is successful, set $w_j \coloneqq 1$.
            Otherwise, set $w_j \coloneqq 0$.
    \end{enumerate}
    \item 
    If $\sum^\ell_{j=1}w_j\ge 2$, $\cC$ outputs $\top$. Otherwise, $\cC$ outputs $\bot$. 
    \end{enumerate}

    \begin{lemma}\label{lem:hyb5}
        
        $|\Pr[\top\gets\rm{Hybrid\:4}]-\Pr[\top\gets\rm{Hybrid\:5}]|\le\negl(\secp)$.
    \end{lemma}

    \begin{proof}[Proof of \cref{lem:hyb5}]
        
        Note that the difference between Hybrid 4 and Hybrid 5 lies only in the step \ref{step:enc_hyb5}. 
        
        Let us consider the following security game of IND-CPA security between a challenger $\cC'$ and a QPT adversary $\cB$:
        \begin{enumerate}
            \item The challenger $\cC'$ runs $\sk\gets\mathsf{SKE.}\KeyGen(1^\secp)$.
            \item $\cB$ runs $\sigk\gets\mathsf{MAC.}\KeyGen(1^\secp)$.
            \item $\cB$ simulates $\cA$ in Hybrid 4 by querying to $\mathsf{SKE.}\Enc(\sk,\cdot)$ up to the step \ref{step:test_query_hyb5}. $\cB$ gets $m^*$, where $m^*$ is the challenge message that $\cA$ sends to $\cC$ in the step \ref{step:test_query_hyb5}.
            \item $\cB$ chooses $i^*\gets[t]$. For each $i\in[t]$, $\cB$ does the following: $\cB$
            chooses $x_i,\theta_i\gets\bit^\secp$ and prepares the state $\eta^0_i\coloneqq\ket{x_i\|\theta_i}\bra{x_i\|\theta_i}\otimes\tau_i$ 
            by running $\tau_i\gets\mathsf{MAC.}\Tag(\sigk,m^*\|x_i\|\theta_i)$. 
            $\cB$ also prepares the state $\eta^1_i\coloneqq\ket{0...0}\bra{0...0}$.
            \item $\cB$ sends the states $\bigotimes_{i=1}^t\eta^0_i$ and $\bigotimes_{i=1}^t\eta^1_i$ to $\cC'$.
            \item $\cC'$ chooses $b\gets\bit$ and gets $\bigotimes_{i\in[t]}\ct_{i}$ by
            running $\ct_i\gets\mathsf{SKE.}\Enc(\sk,\eta^0_i)$ if $b=0$
            and $\ct_i\gets\mathsf{SKE.}\Enc(\sk,\eta^1_i)$ if $b=1$
            for each $i\in[t]$.
            $\cC'$ sends $\bigotimes_{i\in[t]}\ct_{i}$ to $\cB$.
            \item $\cB$ generates $\bigotimes_{i\in[t]}\ket{x_i}_{\theta_i}$. $\cB$ simulates the interaction between $\cC$ and $\cA$ from the step \ref{step:send_hyb5} of Hybrid 4 to the last step by using $\bigotimes_{i\in[t]}\ct_{i}\otimes\bigotimes_{i\in[t]}\ket{x_i}_{\theta_i}$ and querying to $\mathsf{SKE.}\Enc(\sk,\cdot)$. If $\cC$ outputs $\top$, $\cB$ sends $b'\coloneqq0$ to $\cC'$. Otherwise, $\cB$ sends $b'\coloneqq1$ to $\cC'$.
            \item $\cC'$ outputs $\top$ if $b=b'$. Otherwise, $\cC'$ outputs $\bot$.
        \end{enumerate}
       Let $\Pr[b'\gets\cB|b\gets\cC']$ be the probability that $\cB$ sends $b'\in\bit$ to $\cC'$ when $\cC'$ chooses $b\in\bit$. It is clear that $\Pr[0\gets\cB|0\gets\cC']=\Pr[\top\gets\mathrm{Hybrid\:4}]$ and $\Pr[0\gets\cB|1\gets\cC']=\Pr[\top\gets\mathrm{Hybrid\:5}]$. Therefore,
        if $|\Pr[\top\gets\mathrm{Hybrid\:4}]-\Pr[\top\gets\mathrm{Hybrid\:5}]|\ge\frac{1}{\poly(\secp)}$ for infinitely many $\secp\in\mathbb{N}$,
        $\cB$ breaks the IND-CPA security.
    \end{proof}

    Let us define Hybrid 6 as follows. The following lemma is straightforward.

    \begin{lemma}\label{lem:hyb6}
        $\Pr[\top\gets{\rm{Hybrid\:5}}]=\Pr[\top\gets{\rm{Hybrid\:6}}]$.
    \end{lemma}

    \paragraph{Hybrid 6}
    \begin{enumerate}
    \item The challenger $\cC$ runs $\sk\gets\mathsf{SKE.}\KeyGen(1^{\lambda})$ and $\sigk\gets\mathsf{MAC.}\KeyGen(1^\secp)$.
    \item The adversary $\cA$ sends $m^*$ to $\cC$, where $\cA$ can make classical queries to the oracle $\mathcal{O}_{\sk,\sigk}$ and does not query $m^*$. Here, $\mathcal{O}_{\sk,\sigk}$ takes a bit string $m$ as input and works as follows:
    \begin{enumerate}
        \item Choose $x,\theta\gets\bit^\secp$ and generate $\ket{x}_{\theta}$.
        \item Run $\tau\gets\mathsf{MAC.}\Tag(\sigk,m\|x\|\theta)$ and $\ct\gets\mathsf{SKE.}\Enc(\sk,\ket{x\|\theta}\bra{x\|\theta}\otimes\tau)$.
        \item Output $\ket{x}\bra{x}_{\theta}\otimes\ct$.
    \end{enumerate}
    \item $\cC$ chooses $i^*\gets[t]$. For each $i\in[t]$, $\cC$ does the following.\label{step:enc_hyb6}
    \begin{enumerate}
        \item Choose $x_i,\theta_i\gets\bit^\secp$ and generate $\ket{x_i}_{\theta_i}$.
        \item \Erase{Run $\tau_i\gets\mathsf{MAC.}\Tag(\sigk,m^*\|x_i\|\theta_i)$.} Run $\ct_i\gets\mathsf{SKE.}\Enc(\sk,\ket{0...0}\bra{0...0})$.
    \end{enumerate} 
    \item $\cC$ sends $\{\ket{x_i}\bra{x_i}_{\theta_i}\otimes\ct_i\}_{i\in[t]}$ to $\cA$.
    \item $\cA$ sends $\xi_{\regR_1,...,\regR_\ell}$, where $\cA$ can make classical queries to the oracle $\mathcal{O}_{\sk,\sigk}$ and does not query $m^*$. 
    Here $\regR_j$ has two registers $\regA_j$ and $\regC_j$ for each $j\in[\ell]$.
    \item For each $j\in[\ell]$, $\cC$ does the followings.
    \begin{enumerate}
            \item Does nothing in this step.
            \item Set $(x'_j,\theta'_j)\coloneqq(x_{i^*},\theta_{i^*})$.
            \item Does nothing in this step.
            \item Project the register $\regA_j$ onto $\ket{x'_j}_{\theta'_j}$. 
            \item If the projection is successful, set $w_j \coloneqq 1$.
            Otherwise, set $w_j \coloneqq 0$.
    \end{enumerate}
    \item 
    If $\sum^\ell_{j=1}w_j\ge 2$, $\cC$ outputs $\top$. Otherwise, $\cC$ outputs $\bot$. 
    \end{enumerate}

    Finally, we construct an adversary that breaks the security of the Wiesner money scheme from $\cA$ of Hybrid 6,
    which concludes our proof of the theorem.

    \begin{lemma}\label{lem:hyb6_negl}
        $\Pr[\top\gets\rm{Hybrid\:6}]\le\negl(\secp)$.
    \end{lemma}

    \begin{proof}[Proof of \cref{lem:hyb6_negl}]
        Let us assume that there exist polynomials $t,\ell$ and a QPT $\cA$ adversary such that $\Pr[\top\gets{\rm{Hybrid\:6}}]\ge\frac{1}{\poly(\secp)}$ for infinitely many $\secp\in\mathbb{N}$. From this $\cA$, we can construct a QPT adversary $\cB$ that breaks the security of the Wiesner money scheme as follows:
        \begin{enumerate}
            \item The challenger $\cC'$ chooses $x,\theta\gets\bit^\secp$ and sends $\ket{x}_{\theta}$ to $\cB$.
            \item $\cB$ runs $\sk\gets\mathsf{SKE.}\KeyGen(1^\secp)$ and $\sigk\gets\mathsf{MAC.}\KeyGen(1^\secp)$.
            \item $\cB$ simulates the interaction between the challenger and $\cA$ in Hybrid 6, where, in the step $\ref{step:enc_hyb6}$, $\cB$ chooses $i^*\gets[t]$ and replaces $\ket{x_{i^*}}_{\theta_{i^*}}$ with $\ket{x}_{\theta}$. Then, $\cB$ gets $\xi_{\regR_1,...\regR_\ell}$ from $\cA$. $\cB$ chooses $j_0,j_1\gets[\ell]$ and outputs the register $\regA_{j_0}$ and $\regA_{j_1}$.
        \end{enumerate}
        The probability that $\cB$ wins is
        \begin{align}
            \Pr[\cB\:\rm{wins}]
            \ge\binom{\ell}{2}^{-1}\Pr[\top\gets{\rm{Hybrid\:6}}]
            \ge\frac{2}{\ell(\ell-1)}\frac{1}{\poly(\secp)}.
        \end{align}
        However, this contradicts the security of the Wiesner money scheme, \cref{lem:Wiesner_money}. Therefore, $\Pr[\top\gets\rm{Hybrid\:5}]\le\negl(\secp)$.
    \end{proof}
By combining \cref{lem:hyb1,lem:hyb2,lem:hyb3,lem:hyb4,lem:hyb5,lem:hyb6,lem:hyb6_negl}, we have
$\Pr[\top\gets\mbox{Hybrid 0}]\le\negl(\secp)$, but it contradicts the assumption that 
$\Pr[\top\gets\mbox{Hybrid 0}]\ge\frac{1}{\poly(\secp)}$ for infinitely many $\secp$.
Therefore we have
$\Pr[\top\gets\mbox{Hybrid 0}]\le\negl(\secp)$.
\end{proof}

\ifnum\anonymous=1
\else
\paragraph{Acknowledgments.}
TM is supported by
JST CREST JPMJCR23I3, 
JST Moonshot R\verb|&|D JPMJMS2061-5-1-1, 
JST FOREST, 
MEXT QLEAP, 
the Grant-in Aid for Transformative Research Areas (A) 21H05183,
and 
the Grant-in-Aid for Scientific Research (A) No.22H00522.
\fi

\ifnum\submission=0
\bibliographystyle{alpha} 
\else
\bibliographystyle{splncs04}
\fi
\bibliography{abbrev3,crypto,reference}

\appendix
\section{Proof of \cref{lem:SKE_for_classical_imply_SKE_for_quantum}}
\label{sec:proof_of_cSKE_to_qSKE}
In this section, we give a proof of \cref{lem:SKE_for_classical_imply_SKE_for_quantum}. To show it, we use the following lemma
whose proof is straightforward. (For example, see \cite{mahadev2020classical}.)

\begin{lemma}[Pauli Mixing]\label{lem:one-time_pad}
    Let $\regA$ be an $n$-qubit register. Then, for any state $\rho_{\regA,\regB}$ on the registers $\regA$ and $\regB$,
    \begin{align}
        \frac{1}{4^n}
        \sum_{x,z\in\bit^n}
        ((X^xZ^z)_\regA)\otimes I_{\regB})\rho_{\regA,\regB} ((Z^zX^x)_\regA\otimes I_{\regB})=\frac{I_{\regA}}{2^n}\otimes\rho_{\regB}.
    \end{align}
    Here, $\rho_\regB\coloneqq\Tr_\regA\rho_{\regA,\regB}$ and $\Tr_\regA$ is a partial trace of the register $\regA$.
\end{lemma}

\begin{proof}[Proof of \cref{lem:SKE_for_classical_imply_SKE_for_quantum}]
    Let $(\KeyGen,\Enc,\Dec)$ be an IND-CPA-secure SKE scheme
    for classical messages that is secure against QPT adversaries that query the encryption oracle classically. 
    From this, we construct an IND-CPA-secure SKE scheme for quantum messages
    $(\KeyGen',\Enc',\Dec')$ as follows:
    \begin{itemize}
        \item $\KeyGen'(1^\secp)\to\sk':$ Run $\sk\gets\KeyGen(1^\secp)$ and output $\sk'\coloneqq\sk$.
        \item $\Enc'(\sk',\rho)\to\ct':$ Parse $\sk'=\sk$. Let $\rho$ be an $n$-qubit state. It does the following:
        \begin{enumerate}
            \item Choose $x,z\gets\bit^n$ and apply $X^xZ^z$ on $\rho$.
            \item Run $\ct\gets\Enc(\sk,x\|z)$.
            \item Output $\ct'\coloneqq (X^xZ^z\rho Z^zX^x)\otimes\ct$.
        \end{enumerate}
        \item $\Dec'(\sk',\ct')\to\rho:$ Parse $\sk'\coloneqq \sk$. Let $\ct'$ be a state on the register $\regC$. The register $\regC$ consists of two registers $\regM$ and $\regB$.
        If $\ct'$ is honestly generated, $\ct'_\regC=(X^xZ^z\rho Z^zX^x)_\regM\otimes\ct_\regB$. 
        \begin{enumerate}
            \item Run $\Dec(\sk,\cdot)$ on the register $\regB$ to get $x'\|z'$.
            \item Apply $Z^{z'}X^{x'}$ on the register $\regM$ and output the register $\regM$.
        \end{enumerate}
    \end{itemize}
    First, we show the correctness. Fix bit strings $x,z\in\bit^n$ and a polynomial $p$.
    Define the set
    \begin{align}
        S_{x,z,p}\coloneqq\left\{\sk:\Pr[x\|z\gets\Dec(\sk,\ct):\ct\gets\Enc(\sk,x\|z)]\ge1-\frac{1}{p(\secp)}\right\}.
    \end{align}
    From the correctness of $(\KeyGen,\Enc,\Dec)$ and the standard average argument, we have
    \begin{align}
        \Pr[\sk\in S_{x,z,p}:\sk\gets\KeyGen(1^\secp)]\ge1-\negl(\secp).\label{eq:almost_correct}
    \end{align}
    Let us fix $\sk\in S_{x,z,p}$. We define $\Enc'(\sk',\cdot)_{x,z}$ as the following CPTP map:
   \begin{enumerate}
       \item 
    Apply $X^xZ^z$ on the input register $\regM$.
    \item 
   Run $\ct\gets\Enc(\sk,x\|z)$, and set the output $\ct$ on the register $\regB$. 
   \item 
   Output the registers $\regM$ and $\regB$. 
   \end{enumerate} 
   It is clear that $\|\Dec'(\sk',\cdot)\circ\Enc'(\sk',\cdot)_{x,z}-\mathsf{id}_\regM\|_\diamond\le\frac{2}{p(\secp)}$ since $\sk\in S_{x,z,p}$. From this and \cref{eq:almost_correct}, 
   $\Exp_{\sk\gets\KeyGen(1^\secp)}\|\Dec'(\sk',\cdot)\circ\Enc'(\sk',\cdot)_{x,z}-\mathsf{id}_\regM\|_\diamond\le\frac{2}{p(\secp)}+\negl(\secp)$ for all $x,z$.
   Therefore we have
    \begin{align}
       & \Exp_{\sk\gets\KeyGen'(1^\secp)}\|\Dec'(\sk',\cdot)\circ\Enc'(\sk',\cdot)-\mathsf{id}_\regM\|_\diamond\\
       &= \Exp_{\sk\gets\KeyGen'(1^\secp)}\|\Dec'(\sk',\cdot)\circ\Exp_{x,z\gets\bit^n}\Enc'(\sk',\cdot)_{x,z}-\mathsf{id}_\regM\|_\diamond\\
       &\le \Exp_{\sk\gets\KeyGen'(1^\secp)}\Exp_{x,z\gets\bit^{n}}\|\Dec'(\sk',\cdot)\circ\Enc'(\sk',\cdot)_{x,z}-\mathsf{id}_\regM\|_\diamond\\
       & \le\frac{2}{p(\secp)}+\negl(\secp),
    \end{align}
    for any $p$,
    which shows the correctness.
    
     Next, we show the security.
     We define Hybrid 0, which is the original security game of IND-CPA security for quantum messages,
    as follows. 
    
    \paragraph{Hybrid 0}
    \begin{enumerate}
    \item The challenger $\cC$ runs $\sk\gets\KeyGen(1^{\lambda})$.
    \item The adversary $\cA$ sends two $n$-qubit registers $\regM_0$ and $\regM_1$ to $\cC$, where $\cA$ can query the oracle $\mathcal{O}_{\sk}$. Here, $\mathcal{O}_{\sk}$ takes $n$-qubit register $\regM$ as input and works as follows:\label{step:challenge_query}
    \begin{enumerate}
        \item Choose $\alpha,\beta\gets\bit^n$ and apply $X^\alpha Z^\beta$ on the register $\regM$. 
        \item Run $\ct\gets\Enc(\sk,\alpha\|\beta)$.
        \item Output the register $\regM$ and $\ct$.
    \end{enumerate}
    \item $\cC$ chooses $b\gets\bit$ and does the following:
    \begin{enumerate}
        \item  Choose $x,z\gets\bit^n$ and apply $X^x Z^z$ on the register $\regM_b$.
        \item Run $\ct_b\gets\Enc(\sk,x\|z)$.
    \end{enumerate}
    \item  $\cC$ sends the register $\regM_b$ and $\ct_b$ to $\cA$.
    \item $\cA$ sends $b'\in\bit$ to $\cC$, where $\cA$ can query the oracle $\mathcal{O}_{\sk}$.
    \item If $b=b'$, $\cC$ outputs $\top$. Otherwise, $\cC$ outputs $\bot$.
    \end{enumerate}

    For the sake of contradiction, assume that there exist a polynomial $p$ and a QPT adversary $\cA$ such that
   $\Pr[\top\gets\mathrm{Hybrid\:0}]\ge\frac{1}{2}+\frac{1}{p(\secp)}$ for infinitely many $\secp$.
   Define the following Hybrid 1.

    \paragraph{Hybrid 1}
    \begin{enumerate}
    \item The challenger $\cC$ runs $\sk\gets\KeyGen(1^{\lambda})$.
    \item The adversary $\cA$ sends two $n$-qubit registers $\regM_0$ and $\regM_1$ to $\cC$, where $\cA$ can query the oracle $\mathcal{O}_{\sk}$. 
    Here, $\mathcal{O}_{\sk}$ takes $n$-qubit register $\regM$ as input and works as follows:
    \begin{enumerate}
        \item Choose $\alpha,\beta\gets\bit^n$ and apply $X^\alpha Z^\beta$ on the register $\regM$. 
        \item Run $\ct\gets\Enc(\sk,\alpha\|\beta)$.
        \item Output the register $\regM$ and $\ct$.
    \end{enumerate}
    \item $\cC$ chooses $b\gets\bit$ and does the following:\label{step:challenge_query_generate}
    \begin{enumerate}
        \item  Choose $x,z\gets\bit^n$ and apply $X^x Z^z$ on the register $\regM_b$.
        \item \Erase{Run $\ct_b\gets\Enc(\sk,x\|z)$.}
        {\color{red}Run $\ct_b\gets\Enc(\sk,0...0)$.}
    \end{enumerate}
    \item  $\cC$ sends the register $\regM_b$ and $\ct_b$ to $\cA$.\label{step:challenge_query_return}
    \item $\cA$ sends $b'\in\bit$ to $\cC$, where $\cA$ can query the oracle $\mathcal{O}_{\sk}$.\label{step:guess_b}
    \item If $b=b'$, $\cC$ outputs $\top$. Otherwise, $\cC$ outputs $\bot$.
    \end{enumerate}

    \begin{lemma}\label{lem:hyb1_classical_SKE_to_quantum_SKE}
        $|\Pr[\top\gets\rm{Hybrid\:0}]-\Pr[\top\gets\rm{Hybrid\:1}]|\le\negl(\secp)$ for any QPT adversary $\cA$.
    \end{lemma}

    \begin{proof}[Proof of \cref{lem:hyb1_classical_SKE_to_quantum_SKE}]
For the sake of contradiction, let us assume that
there exists a polynomial $p$ and a QPT adversary $\cA$ such that
\begin{align}
|\Pr[\top\gets\mbox{Hybrid 0}]-\Pr[\top\gets\mbox{Hybrid 1}]|\ge\frac{1}{p(\secp)}
\label{CSKE_to_QSKE_assumption}
\end{align}
for infinitely many $\secp$.
        From this $\cA$, we construct a QPT adversary $\cB$ that breaks the SKE scheme for classical messages.
         \begin{enumerate}
         \item The challenger $\cC'$ of the SKE scheme for classical messages runs $\sk\gets\KeyGen(1^\secp)$.
            \item $\cB$ simulates $\cA$ in the Hybrid 0 up to the step \ref{step:challenge_query} and gets two $n$-qubit registers $\regM_0$ and $\regM_1$. 
            Here, when $\cA$ queries $\mathcal{O}_\sk$, $\cB$ simulates it by querying $\Enc(\sk,\cdot)$.
            \item $\cB$ chooses $x,z\gets\bit^n$. $\cB$ sends $x\|z$ and $0...0$ to $\cC'$.
            \item $\cC'$ chooses $b\gets\bit$. If $b=0$, $\cC'$ runs $\ct_{0}\gets\Enc(\sk,x\|z)$. If $b=1$, $\cC'$ runs $\ct_1\gets\Enc(\sk,0...0)$. 
            Then, $\cC'$ sends $\ct_b$ to $\cB$.
            \item $\cB$ chooses $b'\gets\bit$.
            $\cB$ applies $X^xZ^z$ on the register $\regM_{b'}$.
            \item
            $\cB$ simulates $\cA$ from step \ref{step:challenge_query_return} to step \ref{step:guess_b}
            on input the register $\regM_{b'}$ and $\ct_b$. 
            \item $\cB$ gets $b''\in\bit$ from $\cA$, and outputs $1$ if and only if $b''=b'$.
            Otherwise, $\cB$ outputs 0.
        \end{enumerate}
It is easy to see that
    $\Pr[1\gets\cB|b=0]=\Pr[\top\gets\mbox{Hybrid 0}]$ and
    $\Pr[1\gets\cB|b=1]=\Pr[\top\gets\mbox{Hybrid 1}]$.
    Therefore, from the assumption, \cref{CSKE_to_QSKE_assumption}, we have
    \begin{align}
    |\Pr[1\gets\cB|b=0]-\Pr[1\gets\cB|b=1]|\ge\frac{1}{p(\secp)}
    \end{align}
    for infinitely many $\secp$, which breaks the security of the SKE scheme for classical messages.
    \end{proof}

    Let us define the following Hybrid 2.
    From \cref{lem:one-time_pad}, the following is straightforward.
    \begin{lemma}\label{lem:hyb2_classical_SKE_to_quantum_SKE}
        $\Pr[\top\gets\rm{Hybrid\:1}]=\Pr[\top\gets\rm{Hybrid\:2}]$.
    \end{lemma}

    \paragraph{Hybrid 2}

    \begin{enumerate}
    \item The challenger $\cC$ runs $\sk\gets\KeyGen(1^{\lambda})$.
    \item The adversary $\cA$ sends two $n$-qubit registers $\regM_0$ and $\regM_1$ to $\cC$, where $\cA$ can query the oracle $\mathcal{O}_{\sk}$. Here, $\mathcal{O}_{\sk}$ takes $n$-qubit register $\regM$ as input and works as follows:
    \begin{enumerate}
        \item Choose $\alpha,\beta\gets\bit^n$ and apply $X^\alpha Z^\beta$ on the register $\regM$. 
        \item Run $\ct\gets\Enc(\sk,\alpha\|\beta)$.
        \item Output the register $\regM$ and $\ct$.
    \end{enumerate}
    \item $\cC$ chooses $b\gets\bit$ and does the following:\label{step:challenge_query_hyb2}
    \begin{enumerate}
        \item  \Erase{Choose $x,z\gets\bit^n$ and apply $X^x Z^z$ on the register $\regM_b$.}
        {\color{red}Set $\regM_b$ to the maximally mixed state $I/2^n$.}\label{step:one-time_pad}
        \item Run $\ct_b\gets\Enc(\sk,0...0)$.
    \end{enumerate}
    \item  $\cC$ sends the register $\regM_b$ and $\ct_b$ to $\cA$.
    \item $\cA$ sends $b'\in\bit$ to $\cC$, where $\cA$ can query the oracle $\mathcal{O}_{\sk}$.
    \item If $b=b'$, $\cC$ outputs $\top$. Otherwise, $\cC$ outputs $\bot$.
    \end{enumerate}

   It is clear that $\Pr[\top\gets\mbox{Hybrid 2}]=\frac{1}{2}$. Therefore,
   from \cref{lem:hyb1_classical_SKE_to_quantum_SKE,lem:hyb2_classical_SKE_to_quantum_SKE}, 
   we conclude that $\Pr[\top\gets\mbox{Hybrid 0}]\le\frac{1}{2}+\negl(\secp)$, but it contradicts the assumption.
\end{proof}
\section{Private-Key Quantum Money}\label{sec:QM_from_UPSG}

In this section, we recall the definition of private-key money and prove \cref{coro:QM_from_UPSG}.

\begin{definition}[Private-Key Quantum Money Schemes \cite{C:JiLiuSon18,STOC:AarChr12}]\label{def:private-key_money}
A private-key quantum money scheme is a set of algorithms 
$(\KeyGen,\Mint,\Ver)$ such that 
\begin{itemize}
\item
$\KeyGen(1^\secp)\to k:$
It is a QPT algorithm that, on input the security parameter $\secp$,
outputs a classical secret key $k$.
\item
$\Mint(k)\to \$_k:$
It is a QPT algorithm that, on input
$k$, outputs an $m$-qubit quantum state $\$_k$.
\item
$\Ver(k,\rho)\to\top/\bot:$
It is a QPT algorithm that, on input $k$ and a quantum
state $\rho$, outputs $\top/\bot$.
\end{itemize}
We require the following correctness and security.

\paragraph{Correctness:}
\begin{eqnarray*}
\Pr[\top\gets\Ver(k,\$_k):k\gets\KeyGen(1^\secp),\$_k\gets\Mint(k)]
\ge1-\negl(\secp).
\end{eqnarray*}

\paragraph{Security:}
For any QPT adversary $\cA$ and any polynomial $t$,
\begin{eqnarray*}
\Pr[\mathsf{Count}(k,\xi)\ge t+1
:k\gets\KeyGen(1^\secp),\$_k^{\otimes t}\gets\Mint(k)^{\otimes t},\xi\gets\cA(\$_k^{\otimes t})]
\le\negl(\secp),
\end{eqnarray*}
where $\xi$ is a quantum state on $\ell$ registers, $\regR_1,...,\regR_\ell$, each of which is of $m$ qubits, and $\$^{\otimes t}_k\gets\Mint(k)^{\otimes t}$
means that the $\Mint$ algorithm is run $t$ times. Here, $\mathsf{Count}(k,\xi)$ is the following QPT algorithm: for each $j\in[\ell]$, it takes the state on $\regR_j$ as input, and runs $\Ver(k,\cdot)$ to get $\top$ or $\bot$. Then, it outputs the total number of $\top$.
\end{definition}

\begin{corollary}\label{coro:QM_from_unclonable_MAC}
    If MACs with unclonable tags exist, then private-key quantum money schemes exist.
\end{corollary}

\begin{proof}[Proof of \cref{coro:QM_from_unclonable_MAC}]
    Let $(\mathsf{MAC.}\KeyGen,\mathsf{MAC.}\Tag,\mathsf{MAC.}\Ver)$ be a MAC with unclonable tags. From this, we construct a private-key quantum money scheme as follows:
    \begin{itemize}
        \item $\mathsf{QM.}\KeyGen(1^\secp)\to k:$ Run $\sigk\gets\mathsf{MAC.}\KeyGen(1^\secp)$ and output $k\coloneqq\sigk$.
        \item $\mathsf{QM.}\Mint(k)\to\$_k:$ Parse $k=\sigk$. Run $\tau\gets\mathsf{MAC.}\Tag(\sigk,0...0)$ and output $\$_k\coloneqq\tau$.
        \item $\mathsf{QM.}\Ver(k,\rho)\to\top/\bot:$ Parse $k=\sigk$. Run $v\gets\mathsf{MAC.}\Ver(\sigk,0...0,\rho)$ and output it.
    \end{itemize}
    The correctness and the security are clear from those of MAC with unclonable tags.
\end{proof}

From \cref{coro:unclonable_MAC_from_UPSG} and \cref{coro:QM_from_unclonable_MAC}, UPSGs imply private-key quantum money scheme, which means \cref{coro:QM_from_UPSG}. Namely, if UPSGs exist, then private-key quantum money schemes exist.

\end{document}